\documentclass[10pt,journal, twocolumn]{IEEEtran}
\usepackage{graphicx,times, amsmath, amsfonts,comment}
\usepackage{amssymb,epstopdf,bm,bbm,amsthm}
\usepackage[noend]{algorithmic}

\usepackage{multirow, mathtools}

\usepackage{epstopdf, soul, color}
\usepackage{algorithm, array}

\usepackage{relsize}
\newcommand{\beq}{\begin{equation}}
\newcommand{\eeq}{\end{equation}}

\newcommand{\figref}[1]{Fig.~\ref{#1}}

\newcommand*{\Scale}[2][4]{\scalebox{#1}{$#2$}}%
\theoremstyle{plain}
\newtheorem{lemmacounter}{Theorem}
\newtheorem{lemma}[lemmacounter]{Lemma}

\begin{document}

\title{Dynamic User Clustering and Power Allocation for Uplink and Downlink Non-Orthogonal Multiple Access (NOMA) Systems}
\author{\IEEEauthorblockN{Md Shipon Ali, Hina Tabassum, and Ekram Hossain}\thanks{The authors are with the Department of Electrical and Computer Engineering, at the University of Manitoba, Canada (Emails: alims@myumanitoba.ca, \{hina.tabassum,ekram.hossain\}@umanitoba.ca). The work was supported by a Discovery Grant from the Natural Sciences and Engineering Research Council of Canada (NSERC).} }
\maketitle

\begin{abstract}
Non-Orthogonal Multiple Access (NOMA) has recently been considered as a key enabling technique for 5G cellular systems. In NOMA, by exploiting the channel gain differences multiple users are multiplexed into transmission power domain and then non-orthogonally scheduled on the same spectrum resources. Successive interference cancellation (SIC) is then applied at the receiver(s) to decode the message signals. In this paper, first we briefly  describe the differences in the working principles of uplink and downlink NOMA transmissions. Then, for both  uplink and downlink NOMA, we formulate a sum-throughput maximization problem in a cell such that the user clustering (i.e., grouping users into a single cluster or multiple  clusters) and power allocations in NOMA cluster(s) can be optimized under  transmission power constraints, minimum rate requirements of the users, and SIC constraints. Due to the combinatorial nature of the formulated mixed integer non-linear programming (MINLP) problem, we solve the problem in two steps, i.e., by first grouping users into clusters and then optimizing their respective power allocations. In particular, we propose a low-complexity sub-optimal user grouping scheme. The proposed  scheme exploits the channel gain differences among users in a NOMA cluster and group them into a single cluster or multiple clusters in order to enhance the sum-throughput  of the system.  For a given set of NOMA clusters, we then derive the optimal power allocation policy  that maximizes the sum throughput per NOMA cluster and in turn maximizes the overall system throughput. Using KKT optimality conditions, closed-form solutions for optimal power allocations are derived for any cluster size, considering  both uplink and downlink NOMA systems. Numerical results compare the performance of NOMA over orthogonal multiple access (OMA) and illustrate the significance of NOMA in various network scenarios. 
\end{abstract}


\section{Introduction}
Recently, non-orthogonal multiple access (NOMA)~\cite{saito2013} has been considered as a promising technique for  fifth generation (5G) and beyond 5G (B5G) cellular networks. The key idea of NOMA  is to simultaneously serve multiple users (ideally all active users in a serving cell) over same radio resources at the expense of minimal inter-user interference. NOMA not only allows serving individual users with higher effective bandwidth but also allows scheduling more users than the number of available resources. In contrast to  conventional orthogonal multiple access (OMA), where every user is served on exclusively allocated radio resources, NOMA superposes the message signals of multiple users in  power domain by exploiting their respective channel gain differences.
Successive interference cancellation (SIC) is then applied at the receivers for multi-user detection and decoding. For example, in downlink NOMA, the base station (BS) schedules different users over  same resources but their respective message signals are transmitted using different power levels. By exploiting the power differences, each user equipment (UE) can apply SIC and in turn decode its desired signal.

\subsection{Existing Research on NOMA}
Recently, numerous research activities have been initiated across the globe to identify the potential gains of NOMA in both the downlink and uplink transmissions. Here we review the most recent and relevant research studies for uplink and downlink NOMA transmissions.

\subsubsection{Downlink NOMA}
The basic concept of NOMA  was  exploited in \cite{saito2013}-\cite{aben2013} for downlink transmissions. The authors in \cite{saito2013}-\cite{aben2013} proposed power domain user multiplexing at the BSs and SIC-based signal reception at UE terminals.  In \cite{ben2013}, the authors discussed various practical challenges of NOMA systems, such as multi-user power allocation and user scheduling schemes,  error propagation in SIC, overall system overhead, user mobility, and the combination of NOMA with Multiple-Input Multiple-Output (MIMO). System-level and link-level simulations in \cite{aben2013} indicated clear benefits of NOMA over OMA in terms of overall system throughput as well as individual user's throughput. In \cite{ding2014}, closed-form expressions for ergodic sum-rate and outage probability   were presented for two users considering static power allocations.

The impact of user pairing  was studied in \cite{ding2016} for a two-user NOMA system. The authors proposed fixed and opportunistic user pairing schemes by statically allocating transmission powers among NOMA users. On the other hand, the impact of power allocation on the fairness of the downlink NOMA was investigated in \cite{tim2015}, considering perfect channel state information (CSI) feedback as well as average CSI feedback. In \cite{ding2015}, a cooperative NOMA system was studied, where the authors advocated the idea of pairing weak channel users with the strong channel users for cooperative data transmission.  
A test-bed for two-user NOMA system was presented  in \cite{saito2015}. The experiments were performed by setting $5.4$ MHz bandwidth for NOMA users. The results were compared with those for a two-user OMA system where each user has a transmission bandwidth of $2.7$ MHz~\cite{saito2015}. The results showed significance of NOMA over OMA in terms of aggregate as well as individual user's throughput.

\subsubsection{Uplink NOMA}
For uplink transmissions, NOMA  was first investigated in \cite{endu2011} where power control was applied at UE transmitter and minimum mean squared error (MMSE)-based SIC decoding was utilized at BS receiver. A joint subcarrier and power allocation problem was studied in \cite{imari2014}. Specifically, a sub-optimal solution was designed to maximize the sum-rate of a NOMA cluster.
Closed-form expressions for sum-throughput and outage probability were derived  for a two-user uplink NOMA system in \cite{zhang2016} assuming static powers  of different users. The authors in \cite{zhang2016} also compared their results with TDMA-based OMA system and concluded that without proper selection of target data rate for each NOMA user, a user can always be in outage. This conclusion  was also mentioned in \cite{ding2014} for downlink NOMA. Apart from these, a robust user scheduling algorithm for uplink NOMA with SC-FDMA was designed in \cite{li2015}, where the distinct  channel gains of different users were exploited to obtain efficient user grouping.

\subsection{Motivation and Contributions}
For both uplink and downlink NOMA systems, efficient user clustering and power allocation among users are the most fundamental design issues. To date, most of the research investigations  have been conducted either for downlink or for uplink scenario while considering two users in the system with fixed power allocations. In particular, \textbf{there is no comprehensive investigation to precisely analyze the differences in  uplink and downlink NOMA systems and their respective impact on the user grouping and power allocation problems}. In this context, this paper focuses on developing efficient user clustering and power allocation solutions for multi-user uplink and downlink NOMA systems. The contributions of this paper are outlined as follows: 

\begin{itemize}
\item We briefly review and describe the differences in the working principles of uplink and downlink NOMA. 

\item For both uplink and downlink NOMA, we formulate cell-throughput maximization problem such that user grouping and power allocations in NOMA cluster(s) can be optimized under  transmission power constraints, minimum rate requirements of the users, and SIC constraints.

\item Due to the combinatorial nature of the formulated mixed integer non-linear programming (MINLP) problem, we  propose a low-complexity sub-optimal user grouping scheme. The proposed scheme exploits the channel gain differences among users in a NOMA cluster and group them either in a single cluster or multiple clusters to enhance the sum throughput  of the uplink and downlink NOMA systems.  

\item For a given set of NOMA clusters, we derive optimal power allocation that maximizes the sum throughput of all users in a cluster and in turn maximizes the overall system throughput. Using KKT optimality conditions,  for both uplink and downlink NOMA, we derive closed-form optimal power allocations for any cluster size.

\item We evaluate the performances of different uplink and downlink NOMA systems using the proposed user grouping and power allocation solutions. Numerical results compare the performance of NOMA over OMA and illustrate the significance of NOMA in various network scenarios. Important guidelines related to the key design factors of NOMA systems are obtained.

\end{itemize}

\subsection{Paper Organization}
The rest of the paper is organized as follows: Section II discusses the fundamentals of downlink and uplink NOMA systems. Section~III presents the system model, assumptions,  and the joint problem formulation for optimal user clustering and power allocation in NOMA systems. Section~IV and Section~V, respectively, discuss the proposed sub-optimal user clustering solution and the optimal power allocation solutions for uplink and downlink NOMA systems.  Section VI evaluates the performance of the proposed solutions numerically and Section~VII concludes the article.

\section{Fundamentals of Uplink and Downlink NOMA}
In this section, we discuss the basic concepts of downlink and uplink NOMA  considering  $m$ users with distinct channel gains in a cluster.  The power domain multiplexing is applied  to superpose multiple signals,  whereas an SIC mechanism is applied at the receiver(s) to decode the superposed signals. 

\subsection{Downlink NOMA}
Let us consider a general $m$-user downlink NOMA system, where all users experience distinct channel gains. In $m$-user downlink NOMA, a single transmitter (i.e., BS) non-orthogonally transmits $m$ different signals over the same radio resources; whereas, all $m$ receivers (i.e., UEs) receive their desired signals along with the interferences caused by the messages of other UEs. To obtain the desired signal, each SIC receiver first decodes the dominant\footnote{Dominant interference refers to the interference which is sufficiently stronger than the receiver's desired signal.} interferences and then subtracts them  from the superposed signal. Therefore, the received signal strength of the interference signals needs to be sufficiently higher in comparison to the desired signal in order to cancel them by SIC processing at the receiver end.
Since each UE receives all signals (desired and interfering signals) over the same radio spectrum (i.e., channel), multiplexing of different signals with different power levels is crucial to diversify each signal and to perform SIC at a given UE receiver.  

In downlink NOMA, the  messages of high channel gain users are transmitted with low power levels  whereas the  messages of low channel gain users are transmitted with high power levels. As such, at a given receiver, the strong interfering signals are mainly due to the information of low channel gain users. 
The weakest channel user (who receives low interferences due to relatively low powers of the  messages of high channel gain users) cannot suppress  any interferences. However, the highest  channel gain user (who receives  strong interferences due to relatively high powers of the messages of low channel gain users) can suppress all interfering signals.

\begin{figure}[h]
\begin{center}
	\includegraphics[width=3.6 in]{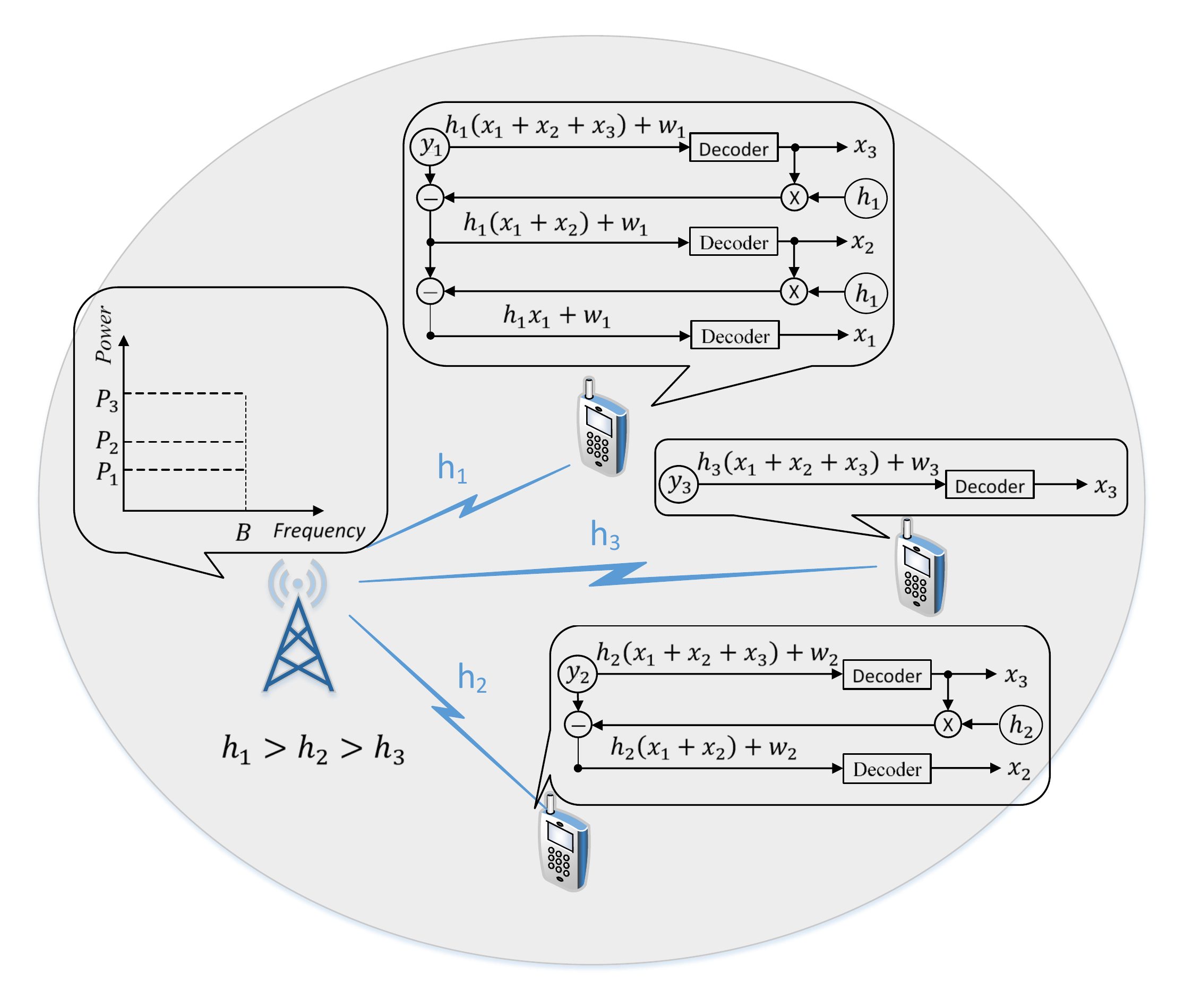}
	\caption{Illustration of a $3$-user downlink NOMA cluster  with SIC at BS.}
	\label{fig:alpha3}
 \end{center}
\end{figure} 

{\bf Illustration:} 
Fig.~\ref{fig:alpha3} illustrates a $3$-user downlink NOMA system where $h_1$, $h_2$, and $h_3$ are the channel gains of $UE_1$, $UE_2$, and $UE_3$, respectively. Also, it is assumed that $x_1$, $x_2$, and $x_3$ are the desired messages of $UE_1$, $UE_2$, and $UE_3$, respectively, while $w_1$, $w_2$, and $w_3$ denote the respective additive white Gaussian noise (AWGN). If $h_1>h_2>h_3$, then $UE_1$ can perform SIC to cancel interference from both $UE_2$ and $UE_3$, whereas $UE_2$ can only cancel interference from $UE_3$. Also, $UE_3$ experiences interference from both $UE_1$ and $UE_2$, but cannot cancel any of them.  Therefore, the achievable throughput for $UE_i, \,\forall\,i=1,2,3$, in a $3$-user downlink NOMA cluster can be expressed as
\begin{align}
\hat R_i = \omega B \log_2\Bigg(1+\frac{P_i \gamma_i}{\sum\limits_{j = 1}^{i-1} P_j\gamma_i+\omega}\Bigg),\,\forall\,i=1,2,3,
\label{ugdl2}
\end{align}
where $\gamma_i=\frac{h_i}{N_0B}$ is the normalized channel gain with $N_0$ being the noise power, $\omega$ is the number of radio channels (e.g., resource blocks in an LTE/LTE-A system) assigned for the cluster,  $B$ is the transmission bandwidth of each channel (e.g., resource block).

To perform SIC, transmission power for each NOMA user needs to be selected properly. If $P_1$, $P_2$, and $P_3$ are the transmission powers for $UE_1$, $UE_2$, and $UE_3$, respectively, then the power allocations need to satisfy the following conditions for efficient SIC at $UE_1$ receiver:
\begin{align}
P_3\gamma_1 - (P_1+P_2)\gamma_1 \geq P_{tol}, \label{dl0}\\
P_2\gamma_1 - P_1\gamma_1 \geq P_{tol}, \quad \enspace
\label{dl1}
\end{align}
where $P_1 + P_2 + P_3 \leq P_t$, $P_t$ is the downlink power budget for this $3$-user NOMA cluster, and $P_{{tol}}$ is the minimum power difference needed to distinguish between the signal to be decoded and the remaining non-decoded message signals. 
\eqref{dl0} and \eqref{dl1} represent the power allocation conditions to cancel interference of $UE_3$ and $UE_2$, respectively, at $UE_1$ receiver. 

From \eqref{dl0} and \eqref{dl1}, it is evident that the transmit power for any user must be the greater than sum transmit power for all users with relatively stronger channel gains. That is, the transmit power for $UE_3$ must be greater than the sum transmit power for $UE_1$ and $UE_2$, while the transmit power for $UE_2$ needs to be greater the transmit power for $UE_1$. Subsequently, the power allocation condition to cancel the interference of $UE_3$ at $UE_2$ receiver can be given as 
\begin{align}
P_3\gamma_2 - (P_1+P_2)\gamma_2 \geq P_{tol}.
\label{dl2}
\end{align}
Note that, since $\gamma_1 > \gamma_2$, \eqref{dl0} is automatically satisfied if \eqref{dl2} holds. Therefore, the necessary power constraints  for efficient SIC in a $3$-user cluster can be given by \eqref{dl1} and \eqref{dl2}.  

Based on the above illustration,
the necessary power constraints for efficient SIC in an $m$-user NOMA cluster can be expressed as follows:
\begin{equation}
P_{i}\gamma_{i-1} - \sum\limits_{j = 1}^{i-1}P_{j}\gamma_{i-1} \geq P_{tol}, \, i = 2,3,\cdots,m.
\label{dm}
\end{equation}
Consequently, an important conclusion about the transmit power of highest channel gain user in a NOMA cluster can be  derived as given in the following.

\begin{lemma}[Maximum Transmit Power for the Highest Channel Gain User in a Downlink NOMA Cluster]
The maximum transmission power allocation to the highest channel gain user in the downlink NOMA cluster must be smaller than $\frac{P_t}{2^{m-1}}$, where $m$ is the number of users in the cluster and $P_t$ is the total transmission power budget for the given NOMA cluster.
\end{lemma}
\begin{proof}
The proof follows by induction. Let us consider a $2$-user downlink NOMA cluster where $\gamma_1$ and $\gamma_2$ are the normalized channel gains of the high and low channel gain users, respectively. As per the SIC constraints in \eqref{dl1}, we have
\begin{align*}
P_2\gamma_1 - P_1\gamma_1 \geq P_{tol}, \quad
\text{and} \quad P_1+P_2 \leq P_{t},
\end{align*}
where $P_1$ and $P_2$ are the allocated powers for high and low channel users, respectively.  The maximum allocated power to the highest channel gain user can thus be derived as
\begin{align*}
P_{1(max)} \leq \frac{P_t-\delta}{2},
\end{align*}
where $\delta=\frac{P_{tol}}{\gamma_1}$ is the minimum power difference needed  for SIC. Note that the value of $\delta$ can be very small when the value of  $\gamma_1$ is very high, which is usually the case. Similarly, for $3$-user downlink NOMA cluster, the maximum allocated transmit powers for second higher and highest channel gain user can be expressed, respectively, as
\begin{align*}
P_{2(max)} &\leq \frac{P_t - \delta}{2}, \\
P_{1(max)} 
&\leq \frac{P_t -\delta}{2^2}-\frac{\delta}{2}.
\end{align*}
Consequently, for an $m$-user cluster, we have
\begin{align*}
P_{1(max)} \leq \frac{P_t-\delta}{2^{m-1}}-\frac{\delta}{2^{m-2}}-\cdots- -\frac{\delta}{2} \approx \frac{P_t}{2^{m-1}}. \end{align*}
This is the same result as given in {\bf Lemma~1}.
\end{proof}

\subsection{Uplink NOMA}
The operation of uplink NOMA is quite different from that of downlink NOMA. In uplink NOMA, multiple transmitters (UEs) non-orthogonally transmit to a single receiver (BS) on the same radio spectrum (i.e., channel). Each UE independently transmits its own signal   at either maximum transmit power or controlled transmit power.  All received signals at the BS  are the desired signals, although they cause interference to each other. Since the transmitters are different, each received signal at SIC receiver (i.e., the BS) experiences distinct channel gain. Note that, to apply SIC  and decode signals at BS, we need to maintain the distinctness among various message signals. As such, conventional transmit power control (typically intended to equalize the received signal powers of all users) may not be feasible in NOMA-based systems.

Let us consider a general $m$-user uplink NOMA system in which $m$ users transmit to a common BS over the same radio channel, at either maximum transmit power or controlled transmit power. The BS receives the superposed message signal of $m$ different users and applies SIC to decode each signal. 
Since the received signal from the highest channel gain user is likely the strongest at the BS; therefore, this signal is decoded first.
Consequently, the highest channel gain user experiences interference from all other users in the cluster.  Then, the signal for second highest channel gain user  is decoded and so on.  As a result, the highest channel gain user experiences  interference from all users and  the lowest channel gain user enjoys interference-free data rate.

\begin{figure}[h]
\begin{center}
	\includegraphics[width=3.6 in]{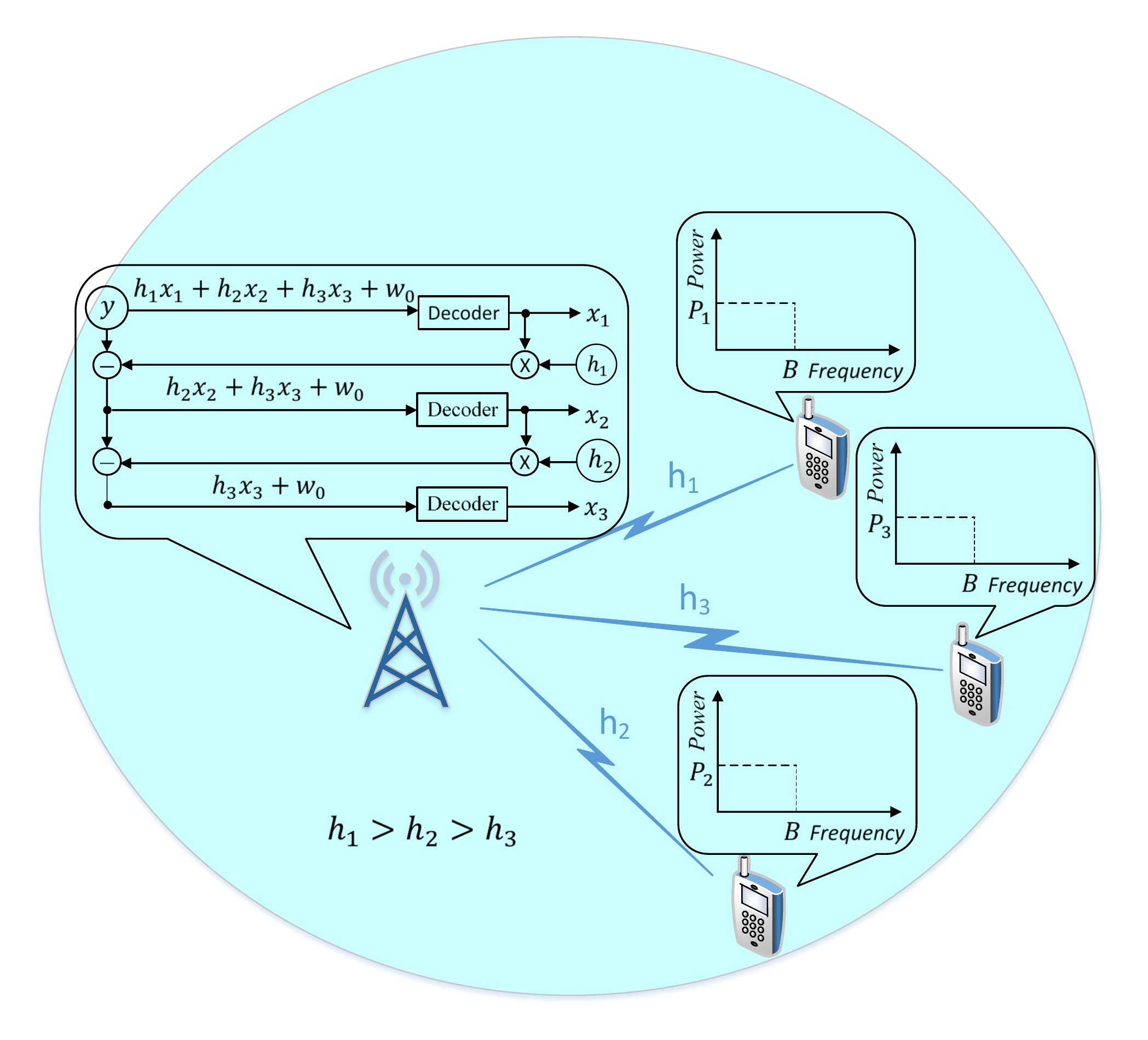}
	\caption{Illustration of a $3$-user uplink NOMA cluster with SIC at UEs.}
	\label{fig:alpha3-up}
 \end{center}
\end{figure}

{\bf Illustration:} 
Fig. \ref{fig:alpha3-up} illustrates a $3$-user uplink NOMA cluster in which $UE_1$, $UE_2$, and $UE_3$ experience channel gains of $h_1$, $h_2$, and $h_3$, respectively, where $h_1>h_2>h_3$.  In uplink NOMA, the user's signal with the highest channel gain   is decoded first at the BS. Thus, the achievable data rate of $UE_1$ depends on the interference from $UE_2$ and $UE_3$, whereas $UE_3$ achieves interference free data rate. 
The data rate of $UE_2$ depends on the interference from $UE_3$.
Consequently, the achievable data rate (or throughput) for $UE_i, \,\forall\,i=1,2,3$, in a $3$-user uplink NOMA cluster can be expressed as  
\begin{align}
\hat R_i = \omega B \log_2\Bigg(1+\frac{P_i \gamma_i}{\sum\limits_{j = i+1}^{3} P_j\gamma_j+\omega}\Bigg),\,\forall\,i=1,2,3,
\label{ugul1}
\end{align}
where $\gamma, \omega,$  and $B$ are defined similarly as in downlink  NOMA. 
If $P_1$, $P_2$, and $P_3$ are the transmission powers of $UE_1$, $UE_2$, and $UE_3$, respectively, then the following conditions need to be satisfied for efficient SIC at BS, i.e.,
\begin{align}
P_1\gamma_1 - P_2\gamma_2 - P_3\gamma_3\geq P_{tol},  \label{u1}\\
P_2\gamma_2 - P_3\gamma_3\geq P_{tol}, \qquad 
\label{u2}
\enspace
\end{align}
where $P_i\leq P^\prime_t, \, \forall \, i$ and $P^\prime_t$ is the maximum transmit power budget of each UE. \eqref{u1} and \eqref{u2} represent the necessary conditions for efficient decoding of $UE_1$ and $UE_2$ signals, respectively, prior to decoding the signal of $UE_3$.  

Based on the above example,
the necessary power constraints for efficient SIC, in an $m$-user uplink NOMA cluster can be expressed as follows:
\begin{equation}
P_{i} \gamma_{i} - \sum\limits_{j=i+1}^{m}P_j\gamma_j \geq  P_{tol}, \, i=1,2,\cdots,(m-1).
\label{um}
\end{equation}

\section{System Model and Problem Formulation}
\subsection{Network Model and Assumptions}
We consider a macro base station (BS) serving $N$ uniformly distributed UEs for uplink as well as downlink. The BS and each of the UEs use a single antenna configuration. The available system bandwidth $B_T$ is divided into frequency resource blocks, each of  bandwidth $B$. That is, the total number of frequency resource blocks are given as $\Omega=B_T/B$.   Users who are non-orthogonally scheduled over the same resource blocks form a NOMA cluster. However, each NOMA cluster operates on orthogonal frequency resource blocks.   The number of users per NOMA cluster is represented by $m$ which ranges from $2 \leq m \leq N$. Also, the resource blocks allocated per cluster are represented by $\omega$, where $1\leq \omega \leq \Omega$.

Provided the range of $m$, the number of clusters can vary between 1 and $N/2$. The maximum BS transmission power budget is $P_T$, the maximum transmission power budget per downlink NOMA cluster is $P_t$,  and the maximum user transmit power is $P^\prime_t$. The normalized channel gain between  $i$-th UE and the BS is represented by $\gamma_i$ which accounts for  both distance-based path-loss and shadowing. The users are sorted according to the descending order of their normalized channel gains as $\gamma_1 > \gamma_2 > \gamma_3 >\cdots> \gamma_{N}$. 

Now, let us define a variable $\beta_{i,j}$ as follows:
\begin{equation}
\beta_{i,j}=
\begin{cases}
1, & \text{if a user $i$ is grouped into cluster $j$}\\
0, & \text{otherwise}
\end{cases}
\end{equation}
where $j=1,2,, \cdots, N/2$ and $i=1,2,\cdots,N$.

\subsection{Problem Formulation: Downlink NOMA}

The joint user clustering (i.e., grouping of users into clusters) and power allocation problem for the throughput maximization in downlink NOMA can be formulated as
\begin{align*}
&\underset{\omega,~\bm \beta, \bm P}
{\text{maximize}}
\Scale[1.1]{\sum\limits_{j = 1}^{N/2} \sum\limits_{i = 1}^{N}
\omega_j
\beta_{i,j}\log_2\Bigg(1+\frac{P_i \gamma_i}{\sum\limits_{k = 1}^{i-1} \beta_{k,j} P_k \gamma_i + \omega_j}\Bigg)}\nonumber\\
&\text{subject to:}
\enspace \bm{\mathrm{C_1:}}~\sum\limits_{j = 1}^{N/2} \sum\limits_{i = 1}^{N}
\beta_{i,j} P_{i} \leq P_{T}, \nonumber\\
&\enspace \bm{\mathrm{C_2:}} \sum\limits_{j = 1}^{N/2} 
\omega_j
\beta_{i,j}\log_2\Bigg(1+\frac{P_i \gamma_i}{\sum\limits_{k = 1}^{i-1} \beta_{k,j} P_k \gamma_i +\omega_j}\Bigg) >  R_i,  \forall\, i,\nonumber\\
& \enspace \bm{\mathrm{C_3:}} \,
\Bigg(\beta_{i,j}P_{i}-\sum\limits_{k = 1}^{i-1} \beta_{k,j} P_k\Bigg) \gamma_{i-1}
\geq P_{tol}, \, \forall\, i,
\nonumber\\
 &\enspace \bm{\mathrm{C_4:}} \,\Bigg(\sum\limits_{j=1}^{N/2} \beta_{i,j} =1, \, \forall \, i\Bigg) \\
&\hspace{30pt}\text{AND}\,\Bigg( \Bigg( \, 2 \leq \sum\limits_{i=1}^{N} \beta_{i,j} \leq N\Bigg)\, \text{OR}\,\Bigg(\sum\limits_{i=1}^{N} \beta_{i,j} =0, \,  \forall \, j \Bigg)\Bigg)
\nonumber\\
& \enspace \bm{\mathrm{C_5:}} \, \, \sum\limits_{j=1}^{N/2} \beta_{i,j} \omega_j \leq \Omega, \,\forall\,i, 
 \enspace \bm{\mathrm{C_6:}} \, \, \omega_j \in \{1,2,\cdots, \Omega\}, \,  \forall \,j,  \\
 & \enspace \bm{\mathrm{C_7:}} \,\, \beta_{i,j} \in \{0,1\}, \, \forall \,i,j,
\end{align*}
where $R_{i}$ is the minimum data rate requirement for $i$-th user. The {\bf Constraint $\mathrm{\bf{C_1}}$} denotes the  total power constraint of the BS, {\bf Constraint $\mathrm{\bf{C_2}}$} ensures the minimum downlink data rate requirements of the users, {\bf Constraint $\mathrm{\bf{C_3}}$} denotes the SIC constraints as discussed in Section~II.A,  {\bf Constraint $\mathrm{\bf{C_4}}$} ensures that one user can be assigned to at most one cluster, while at least two users are grouped into each downlink NOMA cluster, {\bf Constraint $\mathrm{\bf{C_5}}$} provides the total downlink frequency resource constraint. In addition, {\bf Constraint $\mathrm{\bf{C_6}}$} and $\mathrm{\bf{C_7}}$ demonstrate that $\omega$ and $\beta$ are integer variables. 

\subsection{Problem Formulation: Uplink NOMA}
Similarly, the joint user clustering and power allocation problem for throughput maximization of an uplink NOMA system can be formulated as follows:
\begin{align*}\label{P2}
&\underset{\omega,~\bm \beta, \bm P}
{\text{maximize}}
\sum\limits_{j = 1}^{N/2} \sum\limits_{i = 1}^{N}
\omega_j
\beta_{i,j}\Scale[1.2]{\log_2\Bigg(1+\frac{P_i \gamma_i}{\sum\limits_{k = i+1}^{N} \beta_{k,j} P_k \gamma_k +\omega_j}\Bigg)}
\\
&\text{subject to:}
\enspace \bm{\mathrm{C^\prime_1:}} \sum\limits_{j = 1}^{N/2} 
\beta_{i,j} P_{i} \leq P^\prime_t, \forall \, i, \nonumber\\ 
& \enspace\bm{\mathrm{C^\prime_2:}} \sum\limits_{j = 1}^{N/2} 
\omega_j
\beta_{i,j}\log_2\Bigg(1+\frac{P_i \gamma_i}{\sum\limits_{k = i+1}^{N} \beta_{k,j} P_k \gamma_k +\omega_j}\Bigg) \Scale[1]{> R_i^\prime,  \forall \, i,} 
\nonumber \\
&\enspace \bm{\mathrm{C^\prime_3:}} \,\,P_{i} \gamma_{i} \beta_{i,j} - \sum\limits_{k=i+1}^{N}  \beta_{k,j} P_k \gamma_k \geq  P_{tol}, \forall\, i,
\nonumber
\end{align*}
\begin{align*}
& \enspace \bm{\mathrm{C^\prime_4:}} \,\Bigg(\sum\limits_{j=1}^{N/2} \beta_{i,j} =1, \, \forall \, i\Bigg) \\
&\hspace{30pt}\text{AND}\,\Bigg( \Bigg( \, 2 \leq \sum\limits_{i=1}^{N} \beta_{i,j} \leq N\Bigg)\, \text{OR}\,\Bigg(\sum\limits_{i=1}^{N} \beta_{i,j} =0, \,  \forall \, j \Bigg)\Bigg)
\nonumber\\
& \enspace \bm{\mathrm{C^\prime_5:}} \, \, \sum\limits_{j=1}^{N/2} \beta_{i,j} \omega_j \leq \Omega, \,\forall\,i, 
 \enspace \bm{\mathrm{C^\prime_6:}} \, \, \omega_j \in \{1,2,\cdots, \Omega\}, \,  \forall \,j,  \\
 & \enspace \bm{\mathrm{C^\prime_7:}} \,\, \beta_{i,j} \in \{0,1\}, \, \forall \,i,j,
\end{align*}
where $R^\prime_{i}$ is the minimum uplink data rate requirement for $i$-th user. 
{\bf Constraint $\mathrm{\bf{C^\prime_1}}$} ensures the minimum rate requirements of the users, {\bf Constraint $\mathrm{\bf{C^\prime_2}}$} ensures the minimum uplink data rate requirements of the users,  {\bf Constraint $\mathrm{\bf{C^\prime_3}}$} ensures the SIC constraint as discussed in Section~II.B, and
{\bf Constraints $\mathrm{\bf{C^\prime_4}}-\mathrm{\bf{C^\prime_7}}$} are same as defined in Section~III.B.

\subsection{Solution Methodology}
As can be seen, the formulated problems are mixed integer non-linear programming (MINLP) problems whose solution is combinatorial by nature.
Specifically, for throughput maximization, the optimal user clustering solution  requires an exhaustive search to form a NOMA cluster~\cite{li2015}. That is, for every single user, we need to consider all possible combinations of user grouping. For example, let us consider an uplink/downlink NOMA system with $N$  users. In such a system, the  number of possible combinations  for optimal user clustering, can be expressed as follows:
\begin{align*}
\Phi = \mathlarger{\mathlarger{\sum}}_{i = 2}^{N} {N \choose i}. 
\end{align*}
Evidently, the computational complexity of optimal user clustering may not be affordable for practical systems with a large number of active users.
As such, we  resort to solve the problem in two steps, i.e., by developing a less complex solution for grouping users into different NOMA clusters and then optimizing their respective powers  to maximize the sum throughput per cluster. Subsequently, Section~IV details the proposed low-complexity user clustering scheme. Given the user clustering, we derive optimal power allocations for users  in Section~V.

\section{User Clustering in NOMA}
In this section, we propose a low-complexity sub-optimal user clustering scheme for both uplink and downlink NOMA systems. The proposed scheme exploits the channel gain differences among users and aims at  enhancing the sum-throughput  of the considered cell. Prior to user grouping, this scheme relies on selecting a feasible number of  clusters, i.e., decides the number of clusters and in turn the number of users per cluster. Once the number of users in a cluster is decided,  user grouping is performed. In the following, the key concepts of our proposed user paring policies and their algorithmic presentation are detailed.


\subsection{Key Issues for User Clustering in Downlink NOMA}
Let us consider an $m$-user downlink NOMA cluster to which $\omega$ units of resource blocks are allocated. For such a system, the achievable per-user throughput in (\ref{ugdl2}), where $i = 1, 2,3,\cdots,m$, provides necessary insights to group users into a cluster. These insights are discussed below:

\renewcommand{\labelitemi}{$\blacksquare$}
\begin{itemize}
\item After SIC, the throughput of the highest channel gain user in a cluster is not subject to the intra-cluster interference; instead, its throughput depends on its own channel gain and  power. Although the allocated transmit power for the highest channel gain  user is low (as also mentioned in Section~II.A),  its impact on the throughput is minimal. 
Subsequently, if the gain of the highest channel is sufficiently high, then the achievable data rate negligibly depends on the transmission power, unless the power is very low.
Thus it is  beneficial to distribute the high channel gain users in a cell into  different NOMA clusters, as they can significantly contribute to the sum-throughput of a cluster.

\item To increase the throughput of the  users with low channel gains, it is useful to pair them with the high channel gain users. The reason is that the high channel gain users can achieve a higher rate even with the  low power levels while making the large fraction of power available for weak channel users. As such, the key point of our proposed  user clustering in downlink NOMA is to pair the highest channel gain user and the lowest channel gain user into same NOMA cluster, while the second highest channel gain user and the second lowest channel gain user into another NOMA cluster, and so on. 

\item The throughput of the remaining users in a NOMA cluster follows the   same format. That is, the SINR contains same channel in both the denominator and numerator, whereas the transmit power in the numerator is greater than the sum power in the denominator (this is given by the SIC constraint). As such, the throughput of the remaining users in a NOMA cluster depends mainly on the distribution of the transmit power levels.  

\end{itemize}

\renewcommand{\labelitemi}{$\blacksquare$}

\subsection{Key Issues for User Clustering in Uplink NOMA}
Let us consider an $m$-user uplink NOMA cluster to which  $\omega$ units of resources are allocated. For such a system, the achievable per-user throughput in (\ref{ugul1}), where $i = 1, 2,3,\cdots,m$, provides necessary insights to group users. These insights are discussed below: 
\renewcommand{\labelitemi}{$\blacksquare$}
\begin{itemize}
\item In an uplink NOMA cluster, all users' signals experience distinct channel gains. To perform SIC at BS, we need to maintain the distinctness of received signals. As such, the conventional transmit power control may not  be feasible  in a NOMA cluster. Further, contrary to downlink NOMA, 
 the power control at any user doesn't increase the power budget for any other user in a cluster. As such, the ultimate result of  power control is sum-throughput degradation. 
\item  The distinctness among the channels of different users within a NOMA cluster is crucial to minimize inter-user interference and thus to maximize the cluster throughput. 
\item In  uplink NOMA, the highest channel gain user does not interfere to weak channel users (actually his interference is canceled by SIC).  Therefore, this user can transmit with maximum power to achieve a higher throughput. It is thus beneficial to include the high channel gain users transmitting with maximum powers in each NOMA cluster, as they can significantly contribute to the throughput of a cluster.
\end{itemize}

\subsection{User Clustering Algorithm}
Based on the above discussions, let we classify the users into two classes: Class-A and Class-B. The number of users in Class-A, denoted as $\alpha$, have much higher channel gains compared to users in Class-B, i.e., 
\begin{equation*}
\gamma_1, \gamma_2, \cdots, \gamma_\alpha \gg \gamma_{\alpha+1}, \gamma_{\alpha+2}, \cdots, \gamma_N.
\end{equation*}
Then, the proposed sub-optimal user clustering algorithm can be given as in {\bf Algorithm~1} below. 
\\
\rule[.1ex]{\linewidth}{1.5pt}
\Scale[.95]{\textbf{Algorithm 1: User Clustering in Downlink and Uplink NOMA}}
\rule[.8ex]{\linewidth}{.5pt}
\begin{enumerate}
\item[\textbf{1.}] \textbf{Sort users:} $\gamma_1 \geq \gamma_2 \geq \cdots \geq \gamma_\alpha \gg \gamma_{\alpha+1} \geq \gamma_{\alpha+2} \geq \cdots \geq \gamma_N$.
\vspace{6pt}\item[\textbf{2.}] \textbf{Select no. of clusters:} $\text{if}\enspace (\alpha < N/2) \text{ then}\\  \text{\hspace{100pt} number of clusters, }\kappa = \alpha \\
\text{\hspace{95pt} else if}\enspace (\alpha \geq N/2) \text{ then} \\
\text{\hspace{100pt} number of clusters, }\kappa = N/2.$
\vspace{6pt}
\hspace{6pt}\item[\textbf{3.}] 
\textbf{(a) Group users into clusters for  downlink NOMA:} \\
$1$st cluster $=\{\gamma_1, \gamma_{\kappa+1},\gamma_{2\kappa+1}, \cdots,\gamma_N\}$, \\
$2$nd cluster $=\{\gamma_2, \gamma_{\kappa+2}, \gamma_{2\kappa+2},\cdots,\gamma_{N-1}\},\cdots$,\\ $\kappa$-th cluster $=\{\gamma_{\kappa}, \gamma_{2\kappa}, \gamma_{3\kappa},\cdots,\gamma_{N-\kappa-1}\}$.

\vspace{6pt} \textbf{(b) Group users into clusters for uplink NOMA:} \\
$1$st cluster $=\{\gamma_1, \gamma_{\kappa+1},\gamma_{2\kappa+1}, \cdots,\gamma_{N-\kappa-1}\},$\\
$2$nd cluster $=\{\gamma_2, \gamma_{\kappa+2}, \gamma_{2\kappa+2},\cdots,\gamma_{N-\kappa-2}\},\cdots,$ \\
$\kappa$-th cluster $=\{\gamma_{\kappa}, \gamma_{2\kappa}, \gamma_{3\kappa},\cdots,\gamma_{N}\}$.

\vspace{6pt}\item[\textbf{4.}] \textbf{Cluster size:} $\text{if}\enspace (N \textbf{ mod } \kappa == 0)  \text{ then}\\  
\text{\hspace{67pt} uniform cluster size} \\
\text{\hspace{58pt} else if}\enspace (N \textbf{ mod } \kappa \neq 0) \text{ then}\\    
\text{\hspace{67pt} different cluster size.}$
\end{enumerate}
\rule[1.5ex]{\linewidth}{1.5pt}

To illustrate, in Fig.~3 and Fig.~4, we show user grouping  for $2$-user, $3$-user, and $4$-user  NOMA clusters in downlink and uplink, respectively, where the total number of users are taken as 12. 

\begin{figure}[h]
\begin{center}
	\includegraphics[width=3.5 in]{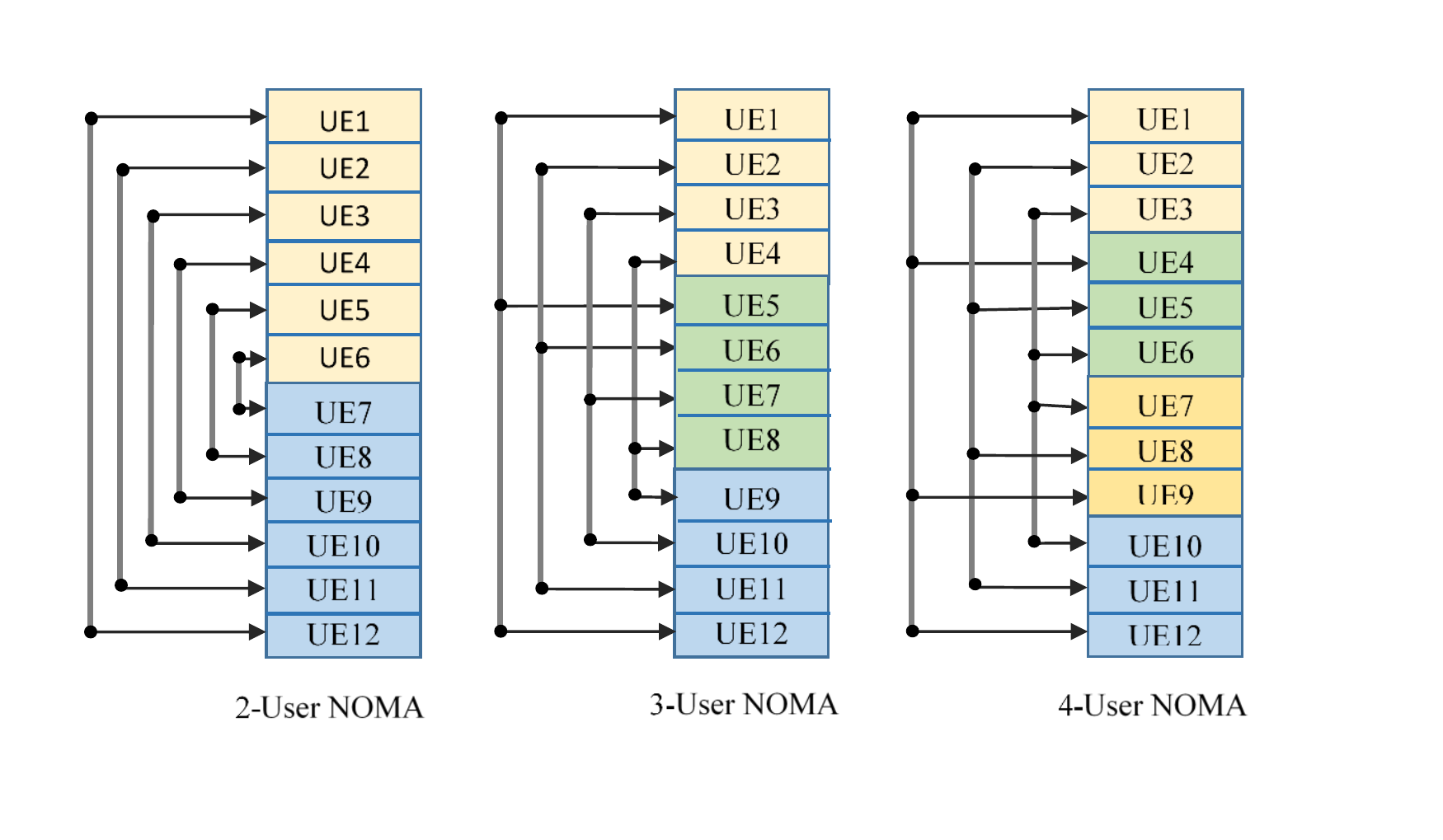}
	\caption{Illustration of $2$-user, $3$-user, and $4$-user NOMA clustering for downlink transmission to 12 active users in a cell.}
	\label{fig:alpha3}
 \end{center}
\end{figure}

\begin{figure}[h]
\begin{center}
	\includegraphics[width=3.5 in]{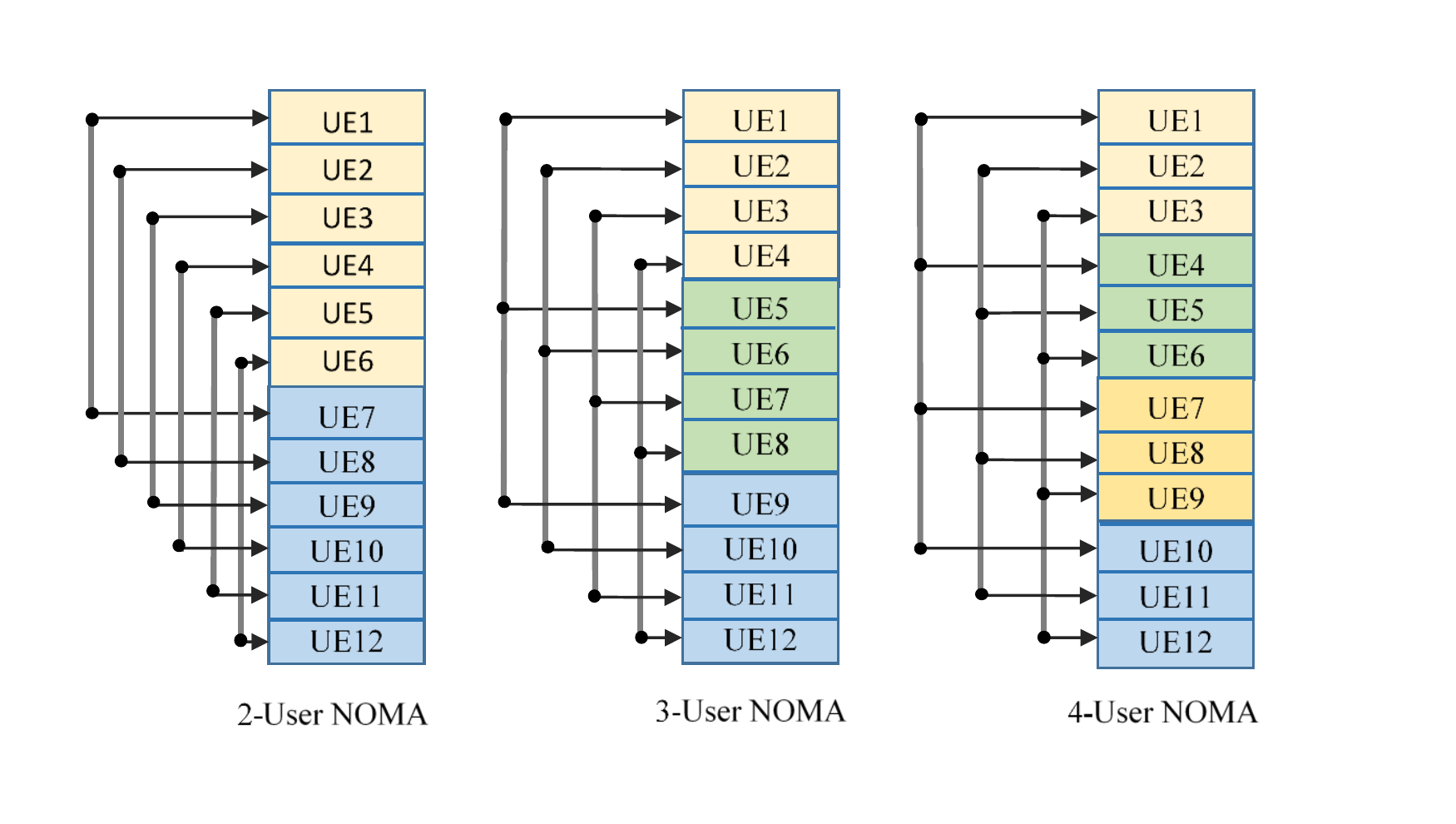}
	\caption{Illustration of $2$-user, $3$-user, and $4$-user NOMA for uplink transmission of 12 active users in a cell.}
	\label{fig:alpha3}
 \end{center}
\end{figure}

\section{Optimal Power Allocations in  NOMA}

Given the NOMA clusters obtained from Section~IV, in this section, we derive optimal power allocations for a NOMA cluster with  $m$ users in the both downlink and uplink transmission scenarios, where $2\leq m \leq N$. Closed-form expressions for the optimal power allocations are derived using KKT optimality conditions. 

\subsection{Downlink NOMA} 
\subsubsection{\textbf{Problem Formulation}}
Let us consider an $m$-user downlink NOMA cluster, where the normalized channel gains of $UE_1, UE_2,\cdots, UE_m$ are assumed as $\gamma_1, \gamma_2,\cdots,\gamma_m$, respectively, and their respective minimum rate requirements are $R_1, R_2,\cdots,R_m$, where $R_i>0$. It is assumed that $\omega$  resource blocks, each of bandwidth $B$, are allocated to the downlink NOMA cluster. If $P_1, P_2,\cdots,P_m$ are the transmission powers for $UE_1, UE_2,\cdots,UE_m$, respectively, then the optimal power allocation problem can be given as
\begin{align*} 
&\quad\underset{P}{\text{max}} \quad \omega B \mathlarger{\mathlarger{\sum}}_{i = 1}^{m}\log_2\Bigg(1+\frac{P_i \gamma_i}{\sum\limits_{j = 1}^{i-1} P_j\gamma_i+\omega}\Bigg)\nonumber 
\\ 
&\text{subject to:}\enspace
\bm{\mathrm{C_1:}}\,\,\sum\limits_{i = 1}^{m}P_i \leq P_t, \nonumber 
\\ 
&\quad\,\, \bm{\mathrm{C_2:}}\,\,\, \omega B\log_2\Bigg(1+\frac{P_i \gamma_i}{\sum\limits_{j = 1}^{i-1} P_j\gamma_i + \omega}\Bigg) \geq R_i,  \forall\, i, \quad\nonumber\\
&\quad\,\, \bm{\mathrm{C_3:}}\,\,\,
P_{i}\gamma_{i-1} - \sum\limits_{j = 1}^{i-1}P_{j}\gamma_{i-1} \geq P_{tol}, \,\forall\, i = 2,3,\cdots,m, 
\end{align*}
where $\sum_{j = 1}^{i-1} P_j\gamma_i$ is the inter-user interference for $i$-th user in downlink NOMA cluster. {\bf Constraint ${\mathrm{\bf{C_1}}}$ } is the total power constraint, {\bf Constraint ${\mathrm{\bf{C_2}}}$ }  is  the minimum rate requirement per user, and {\bf Constraint ${\mathrm{\bf{C_3}}}$ } denotes the SIC constraints. Note that the aforementioned problem  is convex under  {\bf Constraints~$\mathrm{\bf{C_1}}-\mathrm{\bf{C_3}}$}. 

\subsubsection{\textbf{Closed-Form Optimal Power Solution}}

For the aforementioned problem, the Lagrangian can be expressed as: 
\begin{align}
\mathcal{L}&(P,\lambda,\mu,\psi) = \omega B\mathlarger{\mathlarger{‎‎\sum}}_{i = 1}^{m}\log_2\Bigg(1+\frac{P_i\gamma_i}{\sum\limits_{j = 1}^{i-1} P_j\gamma_i+\omega}\Bigg)+ \nonumber\\
&\lambda\Big(P_t-\sum\limits_{i=1}^{m}P_i\Big)+\sum\limits_{i=1}^{m}\mu_i\Big\{P_i\gamma_i-\Big(\sum\limits_{k=1}^{i-1}\ P_k\gamma_i - \omega\Big)\times \nonumber\\
&\Big(\varphi_i-1\Big)\Big\}+ \sum\limits_{i=2}^{m}\psi_i\Big(P_{i}\gamma_{i-1}-\sum\limits_{l=1}^{i}P_l\gamma_{i-1}-P_{tol}\Big),
\label{opd1}
\end{align}
where $\lambda$, $\mu_i$, and $\psi_j$ are the Lagrange multipliers, $\forall \,i = 1,2,3,\cdots,m$ and $\forall \,j = 2,3,4,\cdots,m$. Also, $\varphi_i = 2^{\frac{R_i}{ \omega B}}$. Taking derivatives of (\ref{opd1}) w.r.t. $P_i$, $\lambda$, $\mu_i$,  $\psi_j$, we can write Karush-Kuhn-Tucker (KKT) conditions as follows:
\begin{align}
&\frac{\partial \mathcal{L}}{\partial P_1^\ast}\Scale[1.08]{=\frac{\omega B\gamma_1}{P_1\gamma_1+\omega} - \mathlarger{\mathlarger{\sum}}_{k=2}^m\frac{\omega BP_k\gamma_k^2}{\bigg(\sum\limits_{l=1}^{k}P_l\gamma_k+\omega\bigg)\bigg(\sum\limits_{l^\prime=1}^{k-1}P_{l^\prime}\gamma_k+\omega\bigg)}}\Scale[1.01]{ -\lambda +}\nonumber \\
&\Scale[1]{\mu_1\gamma_1-\sum\limits_{j=2}^{m}(\varphi_j-1)\mu_j\gamma_j  
-\sum\limits_{j=2}^{m}\psi_j\gamma_{j-1} \leq 0,\, \textrm{if} \,P_1^\ast \geq 0},
\label{opd2}
\end{align}
\begin{align}
\frac{\partial \mathcal{L}}{\partial P_i^\ast}&\Scale[1.08]{=\frac{\omega B\gamma_i}{\sum\limits_{j=1}^{i} P_j\gamma_i+\omega} - \mathlarger{\mathlarger{\sum}}_{k=i+1}^m\frac{\omega BP_k\gamma_k^2}{\bigg(\sum\limits_{l=1}^{k}P_l\gamma_k+\omega\bigg)(\sum\limits_{l^\prime=1}^{k-1}P_{l^\prime}\gamma_k+\omega\bigg)} -} \nonumber \\
&\lambda+\mu_i\gamma_i- \sum\limits_{k=i+1}^{m}(\varphi_k-1)\mu_k\gamma_k+\psi_{i}\gamma_{i-1}- \nonumber \\
&\sum\limits_{j=i+1}^{m}\psi_j\gamma_{j-1} \leq 0 ,\textrm{if} \enspace P_i^\ast \geq 0,\, \forall\,i=2,3,\cdots,m,
\label{opd3}
\end{align}
\begin{align}
\frac{\partial \mathcal{L}}{\partial \lambda^\ast}=P_t-\sum\limits_{i=1}^{m}P_i \geq 0,  \,\text{if}\, \lambda^\ast \geq 0, \qquad \qquad \qquad \qquad \, \, 
\label{opd4}  
\end{align}
\begin{align}
\frac{\partial \mathcal{L}}{\partial \mu_i^\ast}=P_i\gamma_i-\Big(\sum\limits_{j=1}^{i-1}P_j\gamma_i+\omega\Big)\Big(\varphi_i-1\Big)\geq 0, \qquad \quad \,\nonumber \\
\textrm{if} \,\mu_i^\ast \geq 0, \, \forall\, i = 1,2,3,\cdots,m, 
\label{opd5} 
\end{align}
\begin{align}
\frac{\partial \mathcal{L}}{\partial \psi_i^\ast}=P_i\gamma_{i-1}-\sum\limits_{j=1}^{i-1}P_j\gamma_{i-1} - P_{tol}\geq 0, \qquad\,\nonumber \qquad \quad\\
\textrm{if} \,\psi_i^\ast \geq 0, \, \forall\, i = 2,3,4,\cdots,m. 
\label{opd6}
\end{align}
In addition, we have several KKT complementarity conditions whose treatment is detailed in the following.

In an $m$-user cluster, there are $2m$ Lagrange multipliers. Thus there are $2^{2m}$ combinations of Lagrange multipliers that need to be checked  for satisfying the KKT conditions~\cite{tianxi}.  However, checking $2^{2m}$ combinations  is computationally complex. For example, if $m=3,4,\cdots,10$, then the number of combinations are $64, 256,\cdots, 1048576$, respectively. In our problem $P_i > 0,\,\forall\,i = 1,2,3,\cdots,m$; therefore, to obtain a fixed number of solutions for $m$ decision variables we need exactly $m$ equations \cite{chong2008}. Thus, all $2^{2m}$ combinations need not to be checked, rather we need to check only ${2m \choose m}$ combinations. After solving for $2$-, $3$-, $4$-, and $6$-user NOMA clusters, we find that the Lagrange multiplier combinations satisfying KKT conditions are $2, 4, 8, 32$, respectively, thus in general $2^{m-1}$.

The Lagrange multipliers for $m$-user downlink NOMA cluster belong to three sets of constraints. These sets are the total transmit power constraints, minimum  data rate constraints, and SIC constraints, given mathematically as,
$A = \{\lambda\}, \quad B = \{\mu_1, \mu_2, \mu_3, \mu_4,\cdots,\mu_m\}, \quad C = \{\psi_2, \psi_3, \psi_4,\cdots,\psi_m\}$, respectively. Therefore, the solution set is, $S = \{\lambda, \mu_2\,\text{or}\,\psi_2, \mu_3\,\text{or}\,\psi_3, \mu_4\,\text{or}\,\psi_4,\cdots,\mu_m \,\text{or}\,\psi_m\}$. 
For example, for $2$-user downlink NOMA, the satisfied KKT conditions are $S_1 = \{\lambda,\mu_2\}$ and $S_2 = \{\lambda,\psi_2\}$. For $3$-user NOMA, the satisfied KKT conditions are $S_1 = \{\lambda,\mu_2, \mu_3\}$, $S_2 = \{\lambda, \mu_2, \psi_3\}$, $S_3 = \{\lambda,\psi_2, \mu_3\}$, and $S_4 = \{\lambda, \psi_2, \psi_3\}$.
Now let us define two additional sets of Lagrange multipliers, $B^\prime = S-B$ and $C^\prime = S-C$. Then the closed-form solution of optimal power allocation to $m$-user downlink NOMA cluster can be given as in the following.

\begin{lemma}[Optimal Power Allocations  for $m$-User Downlink NOMA Cluster]
The closed-form solution of the optimal power allocation for the highest channel gain user in downlink NOMA cluster can be given as follows: 
\begin{align*}
P_1 = \frac{P_t}{\underset{j\not\in B^\prime}{\prod\limits_{j=2}^{m}}\varphi_j \underset{j\in B^\prime}{\prod\limits_{j=2}^{m}}2} - \underset{j\not\in B^\prime}{\mathlarger{\mathlarger{\sum}}_{j=2}^{m}}\frac{\omega(\varphi_j - 1)}{\gamma_j \underset{k\not\in B^\prime}{\prod\limits_{k=2}^{j}}\varphi_k \underset{k\in B^\prime}{\prod\limits_{k=2}^{j}}2} - \\
\underset{j\not\in C^\prime}{\mathlarger{\mathlarger{\sum}}_{j=2}^{m}}\frac{P_{tol}}{2\gamma_{j-1} \underset{k\not\in B^\prime}{\prod\limits_{k=2}^{j-1}}\varphi_k \underset{k\in B^\prime}{\prod\limits_{k=2}^{j-1}}2}.
\end{align*}
On the other hand, the optimal power allocations for remaining users (except the highest channel gain user) are given as
\begin{align*}
\text{(i) If } i\not\in B^\prime, \enspace P_i = \Bigg[\frac{P_t}{\underset{j\not\in B^\prime}{\prod\limits_{j=i}^{m}}\varphi_j \underset{j\in B^\prime}{\prod\limits_{j=i}^{m}}2} - \underset{j\not\in B^\prime}{\mathlarger{\mathlarger{\sum}}_{j=i}^{m}}\frac{\omega(\varphi_j - 1)}{\gamma_j \underset{k\not\in B^\prime}{\prod\limits_{k=i}^{j}}\varphi_k \underset{k\in B^\prime}{\prod\limits_{k=i}^{j}}2} - \\
\underset{j\not\in C^\prime}{\mathlarger{\mathlarger{\sum}}_{j=i}^{m}}\frac{P_{tol}}{2\gamma_{j-1} \underset{k\not\in B^\prime}{\prod\limits_{k=i}^{j-1}}\varphi_k \underset{k\in B^\prime}{\prod\limits_{k=i}^{j-1}}2} + \frac{\omega}{\gamma_i}\Bigg]\times (\varphi_i - 1).
\end{align*}
\begin{align*}
\text{(ii) If } i\in B^\prime, \enspace P_i = \frac{P_t}{\underset{j\not\in B^\prime}{\prod\limits_{j=i}^{m}}\varphi_j \underset{j\in B^\prime}{\prod\limits_{j=i}^{m}}2} - \underset{j\not\in B^\prime}{\mathlarger{\mathlarger{\sum}}_{j=i}^{m}}\frac{\omega(\varphi_j - 1)}{\gamma_j \underset{k\not\in B^\prime}{\prod\limits_{k=i}^{j}}\varphi_k \underset{k\in B^\prime}{\prod\limits_{k=i}^{j}}2} - \\
\underset{j\not\in C^\prime}{\mathlarger{\mathlarger{\sum}}_{j=i}^{m}}\frac{P_{tol}}{2\gamma_{j-1} \underset{k\not\in B^\prime}{\prod\limits_{k=i}^{j-1}}\varphi_k \underset{k\in B^\prime}{\prod\limits_{k=i}^{j-1}}2} + \frac{P_{tol}}{\gamma_{i-1}}.
\end{align*}

\end{lemma}

\begin{proof}
See \textbf{Appendix A}.
\end{proof}

The optimal transmission powers and the corresponding necessary conditions for $2$-, $3$-, and $4$-user downlink  NOMA clusters are provided in {\bf Table~\ref{dlcs}}.
\begin{table*}[]
\centering
\caption{optimal transmission power and corresponding necessary conditions for $2$-, $3$-, and $4$-user downlink NOMA cluster}
\label{dlcs}
\begin{tabular}{|c|l|l|}
\hline
\begin{tabular}[c]{@{}c@{}}NOMA\\ Cluster\end{tabular} & \multicolumn{1}{c|}{Optimal transmit power}                                                                                                                                                                                                                                                                                                                                                                                                                                                                                                                                                                                                                                                                                                                                                                      & \multicolumn{1}{c|}{Necessary conditions}                                                                                                                                                                                                                                          \\ \hline
\multirow{2}{*}{$2$-user}                                        & \begin{tabular}[c]{@{}l@{}}$P_1= \frac{P_t}{\varphi_2}-\frac{\omega(\varphi_2 - 1)}{\varphi_2\gamma_2}$, \\ $P_2= \frac{P_t(\varphi_2-1)}{\varphi_2}+\frac{\omega(\varphi_2 - 1)}{\varphi_2\gamma_2}$\end{tabular}                                                                                                                                                                                                                                                                                                                                                                                                                                                                                                                                                                                                   & \begin{tabular}[c]{@{}l@{}}$P_i\gamma_i-\Big(\varphi_1-1\Big) \Big(\sum\limits_{j=1}^{i-1}P_j\gamma_i + \omega\Big)> 0,\,\forall\, i=1$,\\ $\Big(P_i-\sum\limits_{j=1}^{i-1}P_j\Big)\gamma_{i-1} - P_{tol}> 0,\,\forall\, i=2$\end{tabular}     \\ \cline{2-3} 
                                                                 & \begin{tabular}[c]{@{}l@{}}$P_1= \frac{P_t}{2}-\frac{P_{tol}}{2\gamma_1}$, \\ $P_2= \frac{P_t}{2}+\frac{P_{tol}}{2\gamma_1}$\end{tabular}                                                                                                                                                                                                                                                                                                                                                                                                                                                                                                                                                                                                                                                                    & $P_i\gamma_i-\Big(\varphi_1-1\Big) \Big(\sum\limits_{j=1}^{i-1}P_j\gamma_i + \omega\Big)> 0,\,\forall\, i=1,2$                                                                                                                                                      \\ \hline
\multirow{4}{*}{$3$-user}                                        & \begin{tabular}[c]{@{}l@{}}$P_1= \frac{P_t}{\varphi_2\varphi_3}-\frac{\omega(\varphi_2 - 1)}{\varphi_2\gamma_2}-\frac{\omega(\varphi_3 - 1)}{\varphi_2\varphi_3\gamma_3}$, \\ $P_2= \frac{P_t(\varphi_2 -1)}{\varphi_2\varphi_3}+\frac{\omega(\varphi_2 - 1)}{\varphi_2\gamma_2} - \frac{\omega(\varphi_2 - 1)(\varphi_3 - 1)}{\varphi_2\varphi_3\gamma_3}$, \\ $P_3= \frac{P_t(\varphi_3 -1)}{\varphi_3}+\frac{\omega(\varphi_3 - 1)}{\varphi_3\gamma_3}$\end{tabular}                                                                                                                                                                                                                                                                                                                           & \begin{tabular}[c]{@{}l@{}}$P_i\gamma_i-\Big(\varphi_1-1\Big) \Big(\sum\limits_{j=1}^{i-1}P_j\gamma_i + \omega\Big)> 0,\,\forall\, i=1$,\\ $\Big(P_i-\sum\limits_{j=1}^{i-1}P_j\Big)\gamma_{i-1} - P_{tol}> 0,\,\forall\, i=2,3$\end{tabular}   \\ \cline{2-3} 
                                                                 & \begin{tabular}[c]{@{}l@{}}$P_1= \frac{P_t}{2\varphi_2}-\frac{\omega(\varphi_2 - 1)}{\varphi_2\gamma_2}-\frac{P_{tol}}{2\varphi_2\gamma_2}$, \\ $P_2= \frac{P_t(\varphi_2 -1)}{2\varphi_2}+\frac{\omega(\varphi_2 - 1)}{\varphi_2\gamma_2} -\frac{P_{tol}(\varphi_2 - 1)}{2\varphi_2\gamma_2}$, \\ $P_3= \frac{P_t}{2}+\frac{P_{tol}}{2\gamma_2}$\end{tabular}                                                                                                                                                                                                                                                                                                                                                                                                                 & \begin{tabular}[c]{@{}l@{}}$P_i\gamma_i-\Big(\varphi_1-1\Big) \Big(\sum\limits_{j=1}^{i-1}P_j\gamma_i + \omega\Big)> 0,\,\forall\, i=1,3$,\\ $\Big(P_i-\sum\limits_{j=1}^{i-1}P_j\Big)\gamma_{i-1} - P_{tol}> 0,\,\forall\, i=2$\end{tabular}   \\ \cline{2-3} 
                                                                 & \begin{tabular}[c]{@{}l@{}}$P_1= \frac{P_t}{2\varphi_3}-\frac{P_{tol}}{2\gamma_1}-\frac{\omega(\varphi_3 - 1)}{2\varphi_3\gamma_3}$, \\ $P_2= \frac{P_t}{2\varphi_3}+\frac{P_{tol}}{2\gamma_1}-\frac{\omega(\varphi_3 - 1)}{2\varphi_3\gamma_3}$, \\ $P_3= \frac{P_t(\varphi_3 - 1)}{\varphi_3}+\frac{\omega(\varphi_3 - 1)}{\varphi_3 \gamma_3}$\end{tabular}                                                                                                                                                                                                                                                                                                                                                                                                                                                          & \begin{tabular}[c]{@{}l@{}}$P_i\gamma_i-\Big(\varphi_1-1\Big) \Big(\sum\limits_{j=1}^{i-1}P_j\gamma_i + \omega\Big)> 0,\,\forall\, i=1,2$,\\ $\Big(P_i-\sum\limits_{j=1}^{i-1}P_j\Big)\gamma_{i-1} - P_{tol}> 0,\,\forall\, i=3$\end{tabular}   \\ \cline{2-3} 
                                                                 & \begin{tabular}[c]{@{}l@{}}$P_1= \frac{P_t}{4}-\frac{P_{tol}}{2\gamma_1}-\frac{P_{tol}}{4\gamma_2}$, \\ $P_2= \frac{P_t}{4}+\frac{P_{tol}}{2\gamma_1}-\frac{P_{tol}}{4\gamma_2}$, \\ $P_3= \frac{P_t}{2}+\frac{P_{tol}}{2\gamma_2}$\end{tabular}                                                                                                                                                                                                                                                                                                                                                                                                                                                                                                                                                      & $P_i\gamma_i-\Big(\varphi_1-1\Big) \Big(\sum\limits_{j=1}^{i-1}P_j\gamma_i + \omega\Big)> 0,\,\forall\, i=1,2,3$                                                                                                                                                    \\ \hline
\multirow{8}{*}{$4$-user}                                        & \begin{tabular}[c]{@{}l@{}}$P_1 = \frac{P_t}{\varphi_2\varphi_3\varphi_4} - \frac{\omega(\varphi_2-1)}{\varphi_2\gamma_2} - \frac{\omega(\varphi_3-1)}{\varphi_2\varphi_3\gamma_3} - \frac{\omega(\varphi_4-1)}{\varphi_2\varphi_3\varphi_4\gamma_4}$,\\ $P_2 = \frac{P_t(\varphi_2-1)}{\varphi_2\varphi_3\varphi_4} + \frac{\omega(\varphi_2-1)}{\varphi_2\gamma_2} - \frac{\omega(\varphi_2-1)(\varphi_3-1)}{\varphi_2\varphi_3\gamma_3} - \frac{\omega(\varphi_2-1)(\varphi_4-1)}{\varphi_2\varphi_3\varphi_4\gamma_4}$, \\ $P_3 = \frac{P_t(\varphi_3-1)}{\varphi_3\varphi_4} + \frac{\omega(\varphi_3-1)}{\varphi_3\gamma_3} - \frac{\omega(\varphi_3-1)(\varphi_4-1)}{\varphi_3\varphi_4\gamma_4}$, \\ $P_4 = \frac{P_t(\varphi_4-1)}{\varphi_4} + \frac{\omega(\varphi_4-1)}{\varphi_4\gamma_4}$\end{tabular} & \begin{tabular}[c]{@{}l@{}}$P_i\gamma_i-\Big(\varphi_1-1\Big) \Big(\sum\limits_{j=1}^{i-1}P_j\gamma_i + \omega\Big)> 0,\,\forall\, i=1$,\\ $\Big(P_i-\sum\limits_{j=1}^{i-1}P_j\Big)\gamma_{i-1} - P_{tol}> 0,\,\forall\, i=2,3,4$\end{tabular} \\ \cline{2-3} 
                                                                 & \begin{tabular}[c]{@{}l@{}}$P_1 = \frac{P_t}{2\varphi_2\varphi_3} - \frac{\omega(\varphi_2-1)}{\varphi_2\gamma_2} - \frac{\omega(\varphi_3-1)}{\varphi_2\varphi_3\gamma_3} - \frac{p_{tol}}{2\varphi_2\varphi_3\gamma_3}$,\\ $P_2 = \frac{P_t(\varphi_2-1)}{2\varphi_2\varphi_3} + \frac{\omega(\varphi_2-1)}{\varphi_2\gamma_2}- \frac{\omega(\varphi_2-1)(\varphi_3-1)}{\varphi_2\varphi_3\gamma_3} - \frac{P_{tol}(\varphi_2-1)}{2\varphi_2\varphi_3\gamma_3}$, \\ $P_3 = \frac{P_t(\varphi_3-1)}{2\varphi_3} + \frac{\omega(\varphi_3-1)}{\varphi_3\gamma_3} - \frac{P_{tol}(\varphi_3-1)}{2\varphi_3\gamma_3}$,\\ $P_4 = \frac{P_t}{2} + \frac{P_{tol}}{2\gamma_3}$\end{tabular}                                                                                                           & \begin{tabular}[c]{@{}l@{}}$P_i\gamma_i-\Big(\varphi_1-1\Big) \Big(\sum\limits_{j=1}^{i-1}P_j\gamma_i + \omega\Big)> 0,\,\forall\, i=1,4$,\\ $\Big(P_i-\sum\limits_{j=1}^{i-1}P_j\Big)\gamma_{i-1} - P_{tol}> 0,\,\forall\, i=2,3$\end{tabular} \\ \cline{2-3} 
                                                                 & \begin{tabular}[c]{@{}l@{}}$P_1 = \frac{P_t}{2\varphi_2\varphi_4} - \frac{\omega(\varphi_2-1)}{\varphi_2\gamma_2} - \frac{(\varphi_4-1)}{2\varphi_2\varphi_4\gamma_4} - \frac{P_{tol}}{2\varphi_2\gamma_2}$, \\ $P_2 = \frac{P_t(\varphi_2-1)}{2\varphi_2\varphi_4} + \frac{\omega(\varphi_2-1)}{\varphi_2\gamma_2}- \frac{\omega(\varphi_2-1)(\varphi_4-1)}{2\varphi_2\varphi_4\gamma_4} - \frac{P_{tol}(\varphi_2 - 1)}{2\varphi_2\gamma_2}$, \\ $P_3 = \frac{P_t}{2\varphi_4} +\frac{P_{tol}}{2\gamma_2} - \frac{\omega(\varphi_4-1)}{2\varphi_4\gamma_4}$,\\ $P_4 = \frac{P_t(\varphi_4 - 1)}{\varphi_4} + \frac{\omega(\varphi_4-1)}{\varphi_4\gamma_4}$\end{tabular}                                                                                                                        & \begin{tabular}[c]{@{}l@{}}$P_i\gamma_i-\Big(\varphi_1-1\Big) \Big(\sum\limits_{j=1}^{i-1}P_j\gamma_i + \omega\Big)> 0,\,\forall\, i=1,3$,\\ $\Big(P_i-\sum\limits_{j=1}^{i-1}P_j\Big)\gamma_{i-1} - P_{tol}> 0,\,\forall\, i=2,4$\end{tabular} \\ \cline{2-3} 
                                                                 & \begin{tabular}[c]{@{}l@{}}$P_1 = \frac{P_t}{2\varphi_3\varphi_4} - \frac{P_{tol}}{2\gamma_1} - \frac{\omega(\varphi_3-1)}{2\varphi_3\gamma_3} - \frac{\omega(\varphi_4-1)}{2\varphi_3\varphi_4\gamma_4}$, \\ $P_2 = \frac{P_t}{2\varphi_3\varphi_4} + \frac{P_{tol}}{2\gamma_1} - \frac{\omega(\varphi_3-1)}{2\varphi_3\gamma_3} - \frac{\omega(\varphi_4-1)}{2\varphi_3\varphi_4\gamma_4}$, \\ $P_3 = \frac{P_t(\varphi_3-1)}{\varphi_3\varphi_4} +\frac{\omega(\varphi_3-1)}{\varphi_3\gamma_3} - \frac{\omega(\varphi_3-1)(\varphi_4-1)}{\varphi_3\varphi_4\gamma_4}$,\\ $P_4 = \frac{P_t(\varphi_4 -1)}{\varphi_4} + \frac{\omega(\varphi_4-1)}{\varphi_4\gamma_4}$\end{tabular}                                                                                                                                                             & \begin{tabular}[c]{@{}l@{}}$P_i\gamma_i-\Big(\varphi_1-1\Big) \Big(\sum\limits_{j=1}^{i-1}P_j\gamma_i + \omega\Big)> 0,\,\forall\, i=1,2$,\\ $\Big(P_i-\sum\limits_{j=1}^{i-1}P_j\Big)\gamma_{i-1} - P_{tol}> 0,\,\forall\, i=3,4$\end{tabular} \\ \cline{2-3} 
                                                                 & \begin{tabular}[c]{@{}l@{}}$P_1 = \frac{P_t}{4\varphi_2} - \frac{\omega(\varphi_2-1)}{\varphi_2\gamma_2} - \frac{P_{tol}}{2\varphi_2\gamma_2} - \frac{P_{tol}}{4\varphi_2\gamma_3}$, \\ $P_2 = \frac{P_t(\varphi_2-1)}{4\varphi_2} + \frac{\omega(\varphi_2-1)}{\varphi_2\gamma_2} - \frac{P_{tol}(\varphi_2-1)}{2\varphi_2\gamma_2} - \frac{P_{tol}(\varphi_2-1)}{4\varphi_2\gamma_3}$, \\ $P_3 = \frac{P_t}{4} +\frac{P_{tol}}{2\gamma_2} - \frac{P_{tol}}{4\gamma_3}$, \\ $P_4 = \frac{P_t}{2} + \frac{P_{tol}}{2\gamma_3}$\end{tabular}                                                                                                                                                                                                                                         & \begin{tabular}[c]{@{}l@{}}$P_i\gamma_i-\Big(\varphi_1-1\Big) \Big(\sum\limits_{j=1}^{i-1}P_j\gamma_i + \omega\Big)> 0,\,\forall\, i=1,3,4$,\\ $\Big(P_i-\sum\limits_{j=1}^{i-1}P_j\Big)\gamma_{i-1} - P_{tol}> 0,\,\forall\, i=2$\end{tabular} \\ \cline{2-3} 
                                                                 & \begin{tabular}[c]{@{}l@{}}$P_1 = \frac{P_t}{4\varphi_3} - \frac{P_{tol}}{2\gamma_1} - \frac{\omega(\varphi_3-1)}{2\varphi_3\gamma_3} - \frac{P_{tol}}{4\varphi_3\gamma_3}$, \\ $P_2 = \frac{P_t}{4\varphi_3} + \frac{P_{tol}}{2\gamma_1} - \frac{\omega(\varphi_3-1)}{2\varphi_3\gamma_3} - \frac{P_{tol}}{4\varphi_3\gamma_3}$, \\ $P_3 = \frac{P_t(\varphi_3-1)}{2\varphi_3} +\frac{\omega(\varphi_3-1)}{\varphi_3\gamma_3} - \frac{P_{tol}(\varphi_3-1)}{2\varphi_3\gamma_3}$, \\ $P_4 = \frac{P_t}{2} + \frac{P_{tol}}{2\gamma_3}$\end{tabular}                                                                                                                                                                                                                                                                          & \begin{tabular}[c]{@{}l@{}}$P_i\gamma_i-\Big(\varphi_1-1\Big) \Big(\sum\limits_{j=1}^{i-1}P_j\gamma_i + \omega\Big)> 0,\,\forall\, i=1,2,4$,\\ $\Big(P_i-\sum\limits_{j=1}^{i-1}P_j\Big)\gamma_{i-1} - P_{tol}> 0,\,\forall\, i=3$\end{tabular} \\ \cline{2-3} 
                                                                 & \begin{tabular}[c]{@{}l@{}}$P_1 = \frac{P_t}{4\varphi_4} - \frac{P_{tol}}{2\gamma_1} - \frac{P_{tol}}{4\gamma_2} - \frac{\omega(\varphi_4-1)}{4\varphi_4\gamma_4}$, \\ $P_2 = \frac{P_t}{4\varphi_4} + \frac{P_{tol}}{2\gamma_1} - \frac{P_{tol}}{4\gamma_2} - \frac{\omega(\varphi_4-1)}{4\varphi_4\gamma_4} $, \\ $P_3 = \frac{P_t}{2\varphi_4} + \frac{P_{tol}}{2\gamma_2} - \frac{\omega(\varphi_4-1)}{2\varphi_4\gamma_4}$, \\ $P_4 = \frac{P_t(\varphi_4 - 1)}{\varphi_4} + \frac{\omega(\varphi_4-1)}{\varphi_4\gamma_4}$\end{tabular}                                                                                                                                                                                                                                                                                             & \begin{tabular}[c]{@{}l@{}}$P_i\gamma_i-\Big(\varphi_1-1\Big) \Big(\sum\limits_{j=1}^{i-1}P_j\gamma_i + \omega\Big)> 0,\,\forall\, i=1,2,3$,\\ $\Big(P_i-\sum\limits_{j=1}^{i-1}P_j\Big)\gamma_{i-1} - P_{tol}> 0,\,\forall\, i=4$\end{tabular} \\ \cline{2-3} 
                                                                 & \begin{tabular}[c]{@{}l@{}}$P_1 = \frac{P_t}{8} - \frac{P_{tol}}{2\gamma_1} - \frac{P_{tol}}{4\gamma_2} - \frac{P_{tol}}{8\gamma_3}$,\\ $P_2 = \frac{P_t}{8} + \frac{P_{tol}}{2\gamma_1} - \frac{P_{tol}}{4\gamma_2} - \frac{P_{tol}}{8\gamma_3}$,\\ $P_3 = \frac{P_t}{4} + \frac{P_{tol}}{2\gamma_2} - \frac{P_{tol}}{4\gamma_3}$, \\ $P_4 = \frac{P_t}{2} + \frac{P_{tol}}{2\gamma_3}$\end{tabular}                                                                                                                                                                                                                                                                                                                                                                                                         & $P_i\gamma_i-\Big(\varphi_1-1\Big) \Big(\sum\limits_{j=1}^{i-1}P_j\gamma_i + \omega\Big)> 0,\,\forall\, i=1,2,3,4$                                                                                                                                                  \\ \hline
\end{tabular}
\end{table*}

\subsection{Uplink NOMA}
\subsubsection{\textbf{Problem Formulation}}
Let us consider an $m$-user uplink NOMA cluster, where the normalized channel gains of $UE_1, UE_2,\cdots, UE_m$ are assumed as $\gamma_1, \gamma_2,\cdots,\gamma_m$, respectively, and their respective minimum rate requirements are $R_1^\prime, R_2^\prime,\cdots,R_m^\prime$, where $R_i^\prime>0$. Let us also consider that the $\omega$ units of  resource blocks are allocated to this $m$-user uplink NOMA cluster, where the bandwidth of each resource block  is $B$ Hz. The problem for optimal power control can then be expressed as follows:
\begin{align*} 
&\enspace \underset{P}{\text{max}}\quad \omega B\mathlarger{\mathlarger{‎‎\sum}}_{i = 1}^{m}\log_2\Bigg(1+\frac{P_i \gamma_i}{\sum\limits_{j = i+1}^{m} P_j\gamma_j + \omega}\Bigg) \nonumber\\
&\text{subject to:} \quad \bm{\mathrm{C_1^\prime:}}\,\, P_i \leq P_t^\prime,\,\forall \,  i=1,2,\cdots,m, \nonumber\\
& \quad\bm{\mathrm{C_2^\prime:}}\,\, \omega B\log_2\Bigg(1+\frac{P_i \gamma_i}{\sum\limits_{j = i+1}^{m} P_j\gamma_j + \omega}\Bigg) \Scale[.96]{\geq R_i^\prime,\,\forall i=1,2,\cdots,m,} \nonumber\\
&\quad \bm{\mathrm{C_3^\prime:}}\,\, P_{i} \gamma_{i} - \sum\limits_{j=i+1}^{m-1}P_j\gamma_j \geq P_{tol} ,\,\forall\, i=1,2,\cdots,m-1,
\end{align*}
where $\sum_{j = i+1}^{m} P_j\gamma_j$ is the inter-user interference for $i$-th user in uplink NOMA cluster, and $P_t^\prime$ is the uplink maximum transmission power budget for each user. Note that the aforementioned optimization problem is also convex under the \textbf{Constraints $\mathrm{\bf{C_1^\prime}}-\mathrm{\bf{C_3^\prime}}$}. 

\subsubsection{\textbf{Closed-Form Optimal Power Solution}}
The Lagrange function for the above problem  can then be expressed as 
\begin{align}
\mathcal{L}&(P,\lambda,\mu,\psi)=\omega B\mathlarger{\mathlarger{\sum}}_{i = 1}^{m}\log_2\Bigg(1+\frac{P_i \gamma_i}{\sum\limits_{j = i+1}^{m} P_j\gamma_j + \omega}\Bigg)+ \nonumber\\
&\sum\limits_{i=1}^{m}\lambda_i\Big(P_t^\prime-P_i\Big)+\sum\limits_{i=1}^{m}\mu_i\Big(P_i\gamma_i-\sum\limits_{j=i+1}^{m}\phi_i P_j\gamma_j - \phi_i\omega\Big)+ \nonumber\\
&\sum\limits_{i=1}^{m-1}\psi_i\Big(P_i\gamma_i-\sum\limits_{j=i+1}^{m}P_j\gamma_j-P_{tol}\Big),
\label{opu1}
\end{align}
where $\phi_i = \Big(2^{\frac{R_i^\prime}{\omega B}}-1\Big)$, and $\lambda_i$, $\mu_i$, and $\psi_i$, are the Lagrange multipliers. Taking derivatives of equation (\ref{opu1}) w.r.t. $P_i$, $\lambda_i$, $\mu_i$, and $\psi_i$, we obtain
\begin{align}
\frac{\partial \mathcal{L}}{\partial P_i}=&\frac{\omega B\gamma_i}{\sum\limits_{j = 1}^{m}P_j\gamma_j+\omega} -\lambda_i + \mu_i\gamma_i - \sum\limits_{k=1}^{i-1}\phi_k\mu_k\gamma_i + \gamma_i\psi_i -\nonumber\\
&\sum\limits_{l=1}^{i-1}\psi_l\gamma_l\leq 0, \, \textrm{if} \, P_i^\ast \geq 0,\, \forall \, i = 1,2,\cdots,m-1,
\label{opu2}
\end{align}
\begin{align}
\frac{\partial \mathcal{L}}{\partial P_m}=&\frac{\omega B\gamma_m}{\sum\limits_{j = 1}^{m}P_j\gamma_j+\omega} -\lambda_m + \mu_m\gamma_m - \sum\limits_{k=1}^{m-1}\phi_k\mu_k\gamma_m - \,\, \, \nonumber\\
&\sum\limits_{l=1}^{m}\psi_l\gamma_l\leq 0, \,
\textrm{if} \, P_m^\ast \geq 0,
\label{opu3}
\end{align}
\begin{align}
\frac{\partial \mathcal{L}}{\partial \lambda_i^\ast}=P_t^\prime-P_i \geq 0, \, \textrm{if} \, \lambda_i^\ast \geq 0,\qquad \, \forall \, i = 1,2,\cdots,m, \label{opu4}\\
\frac{\partial \mathcal{L}}{\partial \mu_i^\ast}=P_i\gamma_i-\sum\limits_{j = i+1}^{m}\phi_i P_j\gamma_j- \phi_i \omega\geq 0, \, \textrm{if}\, \mu_i^\ast \geq 0,\enspace \,\, \nonumber \label{opu5}\\
\forall \, i = 1,2,\cdots,m, \\
\frac{\partial \mathcal{L}}{\partial \psi_i^\ast}=P_i \gamma_i - \sum\limits_{j = i+1}^{m}P_j\gamma_j - P_{tol}\geq 0, \, \textrm{if}\, \psi_i^\ast \geq 0, \quad \enspace  \nonumber\\
\forall \, i = 1,2,\cdots,m-1.
\label{opu6}
\end{align}

\renewcommand{\labelitemi}{$\blacksquare$}
In an $m$-user  cluster, there are $(3m - 1)$ Lagrange multipliers, thus there are $2^{3m-1}$ combinations of Lagrange multipliers. Each combination needs to be checked whether it satisfies the KKT conditions or not~\cite{tianxi}.  However, checking $2^{3m-1}$ combinations is computationally complex. 
For example, if $m=3,4,\cdots,10$, then the number of combinations are $256, 2048,\cdots, 536870912$, respectively. However, in our problem $P_i > 0,\,\forall\,i = 1,2,3,...,m$; therefore, we do not need to check all the combinations of Lagrange multipliers.
To obtain a fixed number of solutions for $m$ decision variables, we need exactly $m$ equations \cite{chong2008}. Thus,  only ${3m-1 \choose m}$ combinations need to be checked. After solving for $2$-, $3$-, $4$-, and $6$-user NOMA clusters, we find that the Lagrange multiplier combinations satisfying KKT conditions are $3$ for all cases.

Note that $(3m-1)$ Lagrange multipliers belong to three sets of constraints. These sets are total transmit power constraints, minimum  data rate constraints, and SIC constraints, given mathematically as,
$A = \{\lambda_1, \lambda_2, \lambda_3,\cdots,\lambda_m\}, \quad B = \{\mu_1, \mu_2, \mu_3,,...,\mu_m\}$, and $\quad C = \{\psi_1, \psi_2, \psi_3,...,\psi_{m-1}\}$, respectively. Therefore, the solution set is, $S = \{\lambda_1, \lambda_2\,\lambda_3, ...,\lambda_{m-1}, \lambda_m \,\text{or}\,\mu_{m-1}\,\text{or}\,\psi_{m-1}\}$. For example, for a $3$-user uplink NOMA cluster, the satisfied KKT conditions are $S_1 = \{\lambda_1,\lambda_2, \lambda_3\}$, $S_2 = \{\lambda, \lambda_2, \mu_2\}$, and $S_3 = \{\lambda,\lambda_2, \psi_2\}$. Now let us define three additional sets of Lagrange multipliers as, $A^\prime = S-A$, $B^\prime = S-B$, and $C^\prime = S-C$.  Then the closed-form solution of optimal power allocation to $m$-user uplink NOMA cluster can be given as in the following lemma.

\begin{lemma}[Optimal Power Allocations  for $m$-User Uplink NOMA Cluster]
The closed-form solutions of the optimal power allocations in an $m$-user uplink NOMA cluster can be given as follows: 
\begin{align*}
\text{(i) If ($A^\prime == \{\varnothing\}$), ($B^\prime ==  B$), and ($C^\prime == C$)}, \qquad \qquad \qquad \\
P_i = P_t^\prime, \quad \forall \, i,  \qquad \qquad
\end{align*}
\begin{align*}
\text{(ii) If ($A^\prime \neq \{\varnothing\}$), ($B^\prime \neq B$), and ($C^\prime == C$)}, \qquad \qquad \qquad \enspace \,\\
P_i = P_t^\prime, \quad \forall \, i = 1,2,\cdots,m-1,  \qquad \\
P_m = \frac{P_t^\prime \gamma_{m-1}}{\phi_{m-1}\gamma_m} - \frac{\omega}{\gamma_m},  \qquad \qquad \qquad \, \,
\end{align*}
\begin{align*}
\text{(iii) If ($A^\prime \neq \{\varnothing\}$), ($B^\prime == B$), and ($C^\prime \neq C$)}, \qquad \qquad \qquad \\
P_i = P_t^\prime, \quad \forall \, i = 1,2,3,\cdots,m-1,  \qquad \\
P_m = \frac{P_t^\prime \gamma_{m-1}}{\gamma_m} - \frac{P_{tol}}{\gamma_m}.  \qquad \qquad \qquad \, \,
\end{align*}
\end{lemma}

\begin{proof}
See \textbf{Appendix B}.
\end{proof}

The optimal transmission powers and the corresponding necessary conditions for $2$-, $3$-, and $4$-user uplink  NOMA clusters are provided in {\bf Table~\ref{ulcs}}.
\begin{table*}[ht]
\centering
\caption{Closed-form Solutions for $2$-user, $3$-user, and $4$-user uplink NOMA Cluster}
\label{ulcs}
\begin{tabular}{|c|c|c|}
\hline
\begin{tabular}[c]{@{}c@{}}NOMA\\ Cluster\end{tabular} & Optimal transmission power                                                                                                                                                   & Necessary conditions                                                                                                                                                                                                                                                                                                                 \\ \hline
\multirow{3}{*}{$2$-user}                                      & $P_i = P_t^\prime, \, \forall \, i = 1,2$                                                                                                                          & \begin{tabular}[c]{@{}c@{}}$\textbf{$(C_1^2)$}\quad P_i\gamma_i-\sum\limits_{j = i+1}^{m}\phi_i P_j\gamma_j- \phi_i \omega > 0, \, \forall \, i = 1,2$\\ $\textbf{$(C_2^2)$}\quad P_1\gamma_1 - P_2\gamma_2 - P_{tol} > 0$\end{tabular}                                                                \\ \cline{2-3} 
                                                               & $P_1 = P_t^\prime, \, P_2 = \frac{P_t^\prime\gamma_1}{\phi_1\gamma_2} - \frac{\omega}{\gamma_2}$                                                                        & $\textbf{$(C_1^2)$}\,\forall\, i=2$, $\textbf{$(C_2^2)$}$, and $P_2 < P_t^\prime$                                                                                                                                                                                                                   \\ \cline{2-3} 
                                                               & $P_1 = P_t^\prime, \, P_2 = \frac{P_t^\prime\gamma_1}{\gamma_2} - \frac{P_{tol}}{\gamma_2}$                                                                       & $\textbf{$(C_1^2)$}$, and $P_2 < P_t^\prime$                                                                                                                                                                                                                                                                         \\ \hline
\multirow{3}{*}{$3$-user}                                      & $P_i = P_t^\prime,\,\forall \, i = 1,2,3$                                                                                                                          & \begin{tabular}[c]{@{}c@{}}$\textbf{$(C_1^3)$}\quad P_i\gamma_i-\sum\limits_{j = i+1}^{m}\phi_i P_j\gamma_j- \phi_i \omega > 0, \, \forall \, i = 1,2,3$\\ $\textbf{$(C_2^3)$}\quad P_i \gamma_i - \sum\limits_{j = i+1}^{m}P_j\gamma_j - P_{tol} > 0, \,,\forall \, i = 1,2$\end{tabular}              \\ \cline{2-3} 
                                                               & \begin{tabular}[c]{@{}c@{}}$P_i = P_t^\prime, \,\forall \, i = 1,2$\\ $P_3 = \frac{P_t^\prime\gamma_2}{\phi_2\gamma_3} - \frac{\omega}{\gamma_3}$\end{tabular}          & $\textbf{$(C_1^3)$}\,\forall\, i=1,3$, $\textbf{$(C_2^3)$}$, and $P_3 < P_t^\prime$                                                                                                                                                                                                                 \\ \cline{2-3} 
                                                               & \begin{tabular}[c]{@{}c@{}}$P_i = P_t^\prime, \,\forall \, i = 1,2$\\ $P_3 = \frac{P_t^\prime\gamma_2}{\gamma_3} - \frac{P_{tol}}{\gamma_3}$\end{tabular}         & $\textbf{$(C_1^3)$}$, $\textbf{$(C_2^3)$}\, \forall\, i=1$, and $P_3< P_t^\prime$                                                                                                                                                                                                                   \\ \hline
\multirow{3}{*}{$4$-user}                                      & $P_i = P_t^\prime, \,\forall \, i = 1,2,3,4$                                                                                                                       & \begin{tabular}[c]{@{}c@{}}$\textbf{($C_1^4$)}\quad P_i\gamma_i-\sum\limits_{j = i+1}^{m}\phi_i P_j\gamma_j- \phi_i \omega > 0, \, \forall \, i = 1,2,3,4$\\ $\textbf{($C_2^4$)}\quad P_i \gamma_i - \sum\limits_{j = i+1}^{m}P_j\gamma_j - P_{tol} > 0, \,,\forall \, i = 1,2,3$\end{tabular}         \\ \cline{2-3} 
                                                               & \begin{tabular}[c]{@{}c@{}}$P_i = P_t^\prime, \,\forall \, i = 1,2,3$\\ $P_4 = \frac{P_t^\prime\gamma_3}{\phi_3\gamma_4} - \frac{\omega}{\gamma_4}$\end{tabular}        & $\textbf{($C_1^4$)}\,\forall\, i=1,2,4$, $\textbf{($C_2^4$)}$, and $P_4 < P_t^\prime$                                                                                                                                                                                                               \\ \cline{2-3} 
                                                               & \begin{tabular}[c]{@{}c@{}}$P_i = P_t^\prime, \,\forall \, i = 1,2,3$\\ $P_4 = \frac{P_t^\prime\gamma_3}{\gamma_4} - \frac{P_{tol}}{\gamma_4}$\end{tabular}       & $\textbf{($C_1^4$)}$, $\textbf{($C_2^4$)}\, \forall\, i=1,2$, and $P_4< P_t^\prime$                                                                                                                                                                                                                 \\ \hline
\end{tabular}
\end{table*}

{\bf Remark:} In an uplink NOMA cluster, power control needs to be applied only at the weakest channel gain user. For example, for $4$-user uplink NOMA cluster $UE_1$, $UE_2$, and $UE_3$ transmits with full power, while power control may be needed at $UE_4$ in order to
\begin{itemize}
\item maintain minimum data rate for second weakest channel user ($UE_3$ in $4$-user NOMA cluster), and
\item maintain minimum receive power difference between least two channel gain users ($UE_3$ and $UE_4$ in a $4$-user NOMA cluster) at BS receiver.
\end{itemize}

\begin{table*}[]
\centering
\caption{Performance of $m$-user Downlink NOMA $(m = 2,3,4)$ and OMA systems with 12 users}
\label{dl-compr}
\begin{tabular}{|c|c|c|c|c|c|c|c|c|c|c|c|c|c|c|c|c|}
\hline
\multirow{3}{*}{Case} & \multicolumn{12}{c|}{\multirow{2}{*}{Normalized channel gain (dB)}}                                                                                                     & \multicolumn{4}{c|}{Sum-throughput (Mbps)}                            \\ \cline{14-17} 
                      & \multicolumn{12}{c|}{}                                                                                                                                             & \multicolumn{3}{c|}{NOMA}                       & \multirow{2}{*}{OMA} \\ \cline{2-16}
                      & $\gamma_1$ & $\gamma_2$ & $\gamma_3$ & $\gamma_4$ & $\gamma_5$ & $\gamma_6$ & $\gamma_7$ & $\gamma_8$ & $\gamma_9$ & $\gamma_{10}$ & $\gamma_{11}$ & $\gamma_{12}$ & $4$-UEs        & $3$-UEs        & $2$-UEs       &                      \\ \hline \hline
1                     & 40         & 15         & 14.5       & 14         & 13.5       & 13         & 12.5       & 12         & 11.5       & 11            & 10.5          & 10            & \textbf{12.78} & 11.72          & 10.3          & 8.15                 \\ \hline
2                     & 40         & 39.5       & 15         & 14.5       & 14         & 13.5       & 13         & 12.5       & 12         & 11.5          & 11            & 10.5          & \textbf{18.36} & 16.13          & 13.39         & 9.85                 \\ \hline
3                     & 40         & 39.5       & 39         & 15         & 14.5       & 14         & 13.5       & 13         & 12.5       & 12            & 11.5          & 11            & \textbf{23.74} & 20.37          & 16.38         & 11.51                \\ \hline
4                     & 40         & 39.5       & 39         & 38.5       & 15         & 14.5       & 14         & 13.5       & 13         & 12.5          & 12            & 11.5          & 24.08          & \textbf{24.45} & 19.25         & 13.1                 \\ \hline
5                     & 40         & 39.5       & 39         & 38.5       & 38         & 15         & 14.5       & 14         & 13.5       & 13            & 12.5          & 12            & 24.4           & \textbf{24.62} & 22            & 14.64                \\ \hline
6                     & 40         & 39.5       & 39         & 38.5       & 38         & 37.5       & 15         & 14.5       & 14         & 13.5          & 13            & 12.5          & 24.7           & \textbf{24.77} & 24.65         & 16.13                \\ \hline
7                     & 40         & 39.5       & 39         & 38.5       & 38         & 37.5       & 37         & 15         & 14.5       & 14            & 13.5          & 13            & 24.89          & \textbf{24.91} & 24.71         & 17.56                \\ \hline
8                     & 40         & 39.5       & 39         & 38.5       & 38         & 37.5       & 37         & 36.5       & 15         & 14.5          & 14            & 13.5          & \textbf{25.07} & 25.04          & 24.76         & 18.93                \\ \hline
9                     & 40         & 39.5       & 39         & 38.5       & 38         & 37.5       & 37         & 36.5       & 36         & 15            & 14.5          & 14            & \textbf{25.23} & 25.11          & 24.81         & 20.25                \\ \hline
10                    & 40         & 39.5       & 39         & 38.5       & 38         & 37.5       & 37         & 36.5       & 36         & 35.5          & 15            & 14.5          & \textbf{25.32} & 25.18          & 24.86         & 21.51                \\ \hline
11                    & 40         & 39.5       & 39         & 38.5       & 38         & 37.5       & 37         & 36.5       & 36         & 35.5          & 35            & 15            & \textbf{25.4}  & 25.24          & 24.9          & 22.72                \\ \hline
12                    & 40         & 39.5       & 39         & 38.5       & 38         & 37.5       & 37         & 36.5       & 36         & 35.5          & 35            & 34.5          & \textbf{25.47} & 25.29          & 24.93         & 23.86                \\ \hline
13                    & 40         & 37         & 34         & 31         & 28         & 25         & 22         & 19         & 16         & 13            & 10            & 7             & \textbf{22.84} & 22.03          & 20.11         & 14.24                \\ \hline
14                    & 11         & 10.5       & 10         & 9.5        & 9          & 8.5        & 8          & 7.5        & 7          & 6.5           & 6             & 5.5           & 4.25           & 4.46           & \textbf{4.54} & 4.11                 \\ \hline
\end{tabular}
\end{table*}

\begin{table*}[]
\centering
\caption{Performances of $m$-user uplink NOMA $(m = 2,3,4,6)$ and OMA systems with 12 different channel users}
\label{ul-compr}
\begin{tabular}{|c|c|c|c|c|c|c|c|c|c|c|c|c|c|c|c|c|c|}
\hline
\multirow{3}{*}{Case} & \multicolumn{12}{c|}{\multirow{2}{*}{Normalized channel gain (dB)}}                                                                                                     & \multicolumn{5}{c|}{Sum-throughput (Mbps)}                                \\ \cline{14-18} 
                      & \multicolumn{12}{c|}{}                                                                                                                                             & \multicolumn{4}{c|}{NOMA}                           & \multirow{2}{*}{OMA} \\ \cline{2-17}
                      & $\gamma_1$ & $\gamma_2$ & $\gamma_3$ & $\gamma_4$ & $\gamma_5$ & $\gamma_6$ & $\gamma_7$ & $\gamma_8$ & $\gamma_9$ & $\gamma_{10}$ & $\gamma_{11}$ & $\gamma_{12}$ & $6$-UEs        & $4$-UEs        & $3$-UEs & $2$-UEs &                      \\ \hline \hline
1                     & 40         & 20         & 18.5       & 17         & 15.5       & 14         & 12.5       & 11         & 9.5        & 8             & 6.5           & 5             & \textbf{12.90} & 11.28          & 10.34   & 9.17    & 7.14                 \\ \hline
2                     & 40         & 38.5       & 20         & 18.5       & 17         & 15.5       & 14         & 12.5       & 11         & 9.5           & 8             & 6.5           & \textbf{18.34} & 15.50          & 13.86   & 11.91   & 8.93                 \\ \hline
3                     & 40         & 38.5       & 37         & 20         & 18.5       & 17         & 15.5       & 14         & 12.5       & 11            & 9.5           & 8             & 18.97          & \textbf{19.05} & 16.88   & 14.33   & 10.59                \\ \hline
4                     & 40         & 38.5       & 37         & 35.5       & 20         & 18.5       & 17         & 15.5       & 14         & 12.5          & 11            & 9.5           & \textbf{19.59} & 19.37          & 19.40   & 16.42   & 12.11                \\ \hline
5                     & 40         & 38.5       & 37         & 35.5       & 34         & 20         & 18.5       & 17         & 15.5       & 14            & 12.5          & 11            & \textbf{19.83} & 19.68          & 19.58   & 18.17   & 13.48                \\ \hline
6                     & 40         & 38.5       & 37         & 35.5       & 34         & 32.5       & 20         & 18.5       & 17         & 15.5          & 14            & 12.5          & \textbf{20.06} & 19.98          & 19.76   & 19.58   & 14.69                \\ \hline
7                     & 40         & 38.5       & 37         & 35.5       & 34         & 32.5       & 31         & 20         & 18.5       & 17            & 15.5          & 14            & \textbf{20.17} & 20.08          & 19.93   & 19.65   & 15.75                \\ \hline
8                     & 40         & 38.5       & 37         & 35.5       & 34         & 32.5       & 31         & 29.5       & 20         & 18.5          & 17            & 15.5          & \textbf{20.27} & 20.17          & 20.09   & 19.72   & 16.64                \\ \hline
9                     & 40         & 38.5       & 37         & 35.5       & 34         & 32.5       & 31         & 29.5       & 28         & 20            & 18.5          & 17            & \textbf{20.32} & 20.25          & 20.13   & 19.78   & 17.36                \\ \hline
10                    & 40         & 38.5       & 37         & 35.5       & 34         & 32.5       & 31         & 29.5       & 28         & 26.5          & 20            & 18.5          & \textbf{20.36} & 20.28          & 20.16   & 19.84   & 17.91                \\ \hline
11                    & 40         & 38.5       & 37         & 35.5       & 34         & 32.5       & 31         & 29.5       & 28         & 26.5          & 25            & 20            & \textbf{20.38} & 20.30          & 20.19   & 19.89   & 18.29                \\ \hline
12                    & 40         & 38.5       & 37         & 35.5       & 34         & 32.5       & 31         & 29.5       & 28         & 26.5          & 25            & 23.5          & \textbf{20.40} & 20.32          & 20.22   & 19.92   & 18.49                \\ \hline
13                    & 40         & 37         & 34         & 31         & 28         & 25         & 22         & 19         & 16         & 13            & 10            & 7             & \textbf{18.65} & 18.35          & 17.96   & 16.93   & 12.89                \\ \hline
14                    & 20         & 18.5       & 17         & 15.5       & 14         & 12.5       & 11         & 9.5        & 8          & 6.5           & 5             & 3.5           & \textbf{6.47}  & 6.41           & 6.33    & 6.11    & 5.23                 \\ \hline
\end{tabular}
\end{table*}

\section{Numerical Results and Discussions}
In this section, we investigate the throughput performances of the downlink and uplink NOMA systems, using our proposed user grouping and optimal power allocation solutions. In our simulations, 2, 3, 4, and 6 units of resource blocks are allocated for $2$-, $3$-, $4$-, and $6$-user NOMA clusters, respectively. Both uplink and downlink NOMA systems are also compared with OFDMA-based LTE/LTE-Advanced systems. In addition, the total downlink transmission power is uniformly allocated among the available resource blocks. The major simulation parameters are shown in  \textbf{Table~\ref{prmeters}}.

\begin{table} [H]
\centering
\caption{Simulation parameters for downlink NOMA and uplink NOMA}
\label{prmeters}
\begin{tabular}{|c|c|}
\hline
Parameter         & Value                          \\
\hline \hline
System effective bandwidth      & $20$ MHz                       \\
\hline
Bandwidth of a resource block, $B$ & $180$ kHz                             \\
\hline 
Number of available resource units & $100$                              \\
\hline 
Downlink Transmit power budget, $P_T$     & $46$ dBm                       \\
\hline
Uplink Transmit power budget, $P_t^\prime$       & $24$ dBm                       \\
\hline
SIC receiver's detection threshold, $P_{tol}$    & $10$ dBm  \\
\hline
Number of transmit antenna at both of BS and UE end    & $1$  \\
\hline
Number of receive antenna at both of BS and UE end   & $1$  \\
\hline                      
\end{tabular}
\end{table}

\subsection{Downlink NOMA}
In this subsection, we compare the performance of NOMA with OFDMA in terms of sum-throughput and individual users' throughput. Further, we compare the overall throughput performance of $2$-user, $3$-user, and $4$-user downlink NOMA systems by considering $12$ active downlink users.
\\
\subsubsection{\textbf{Throughput Comparison between NOMA and OMA Systems}}
\begin{figure}[h]
\begin{center}
	\includegraphics[width=3.5 in]{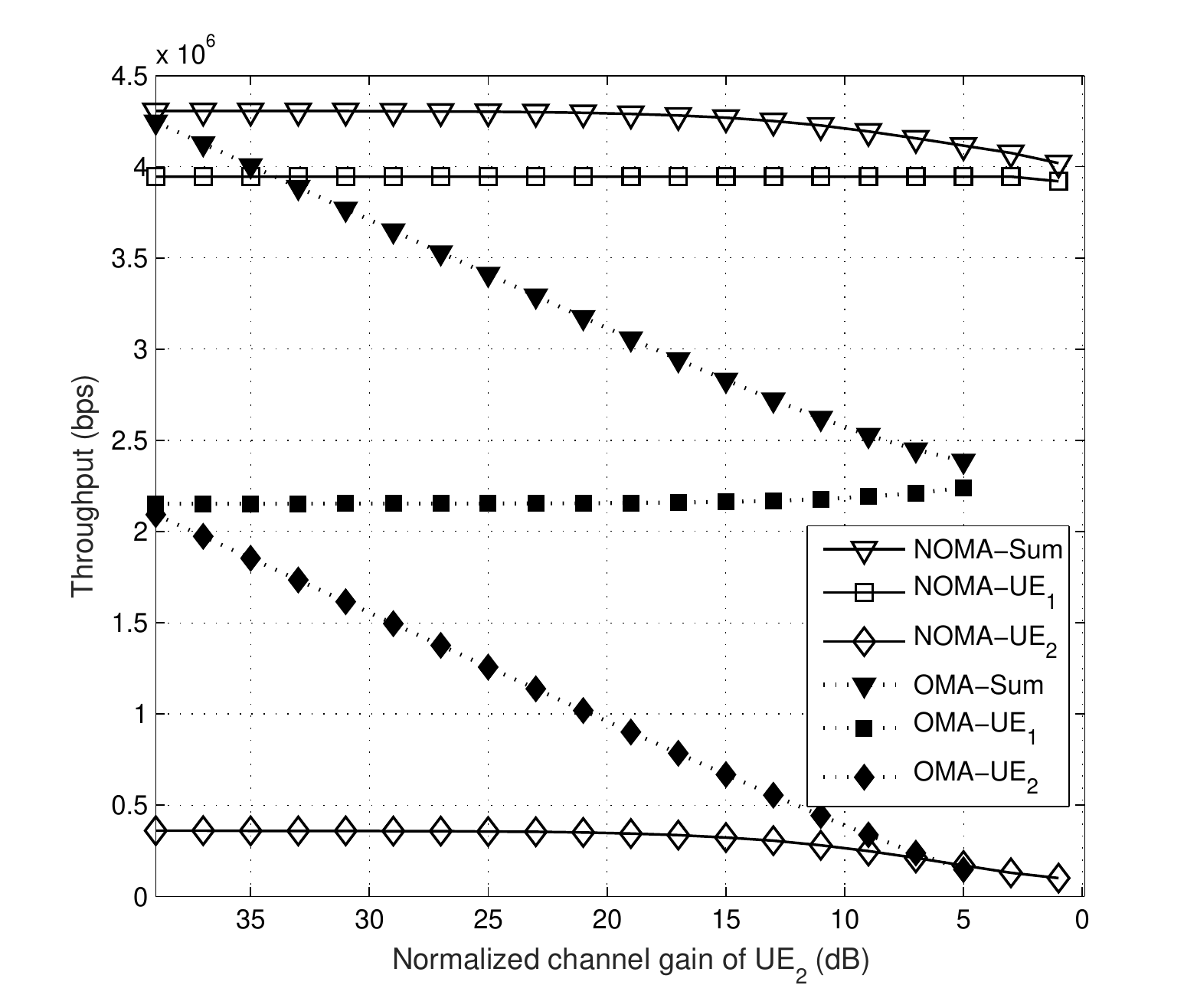}
	\caption{Throughput performance of $2$-user downlink NOMA and OMA systems assuming 100 Kbps minimum data rate. Normalized channel gain of $UE_1$ is 40 dB.}
	\label{fig:1}
 \end{center}
\end{figure}
\begin{figure}[h]
\begin{center}
	\includegraphics[width=3.5 in]{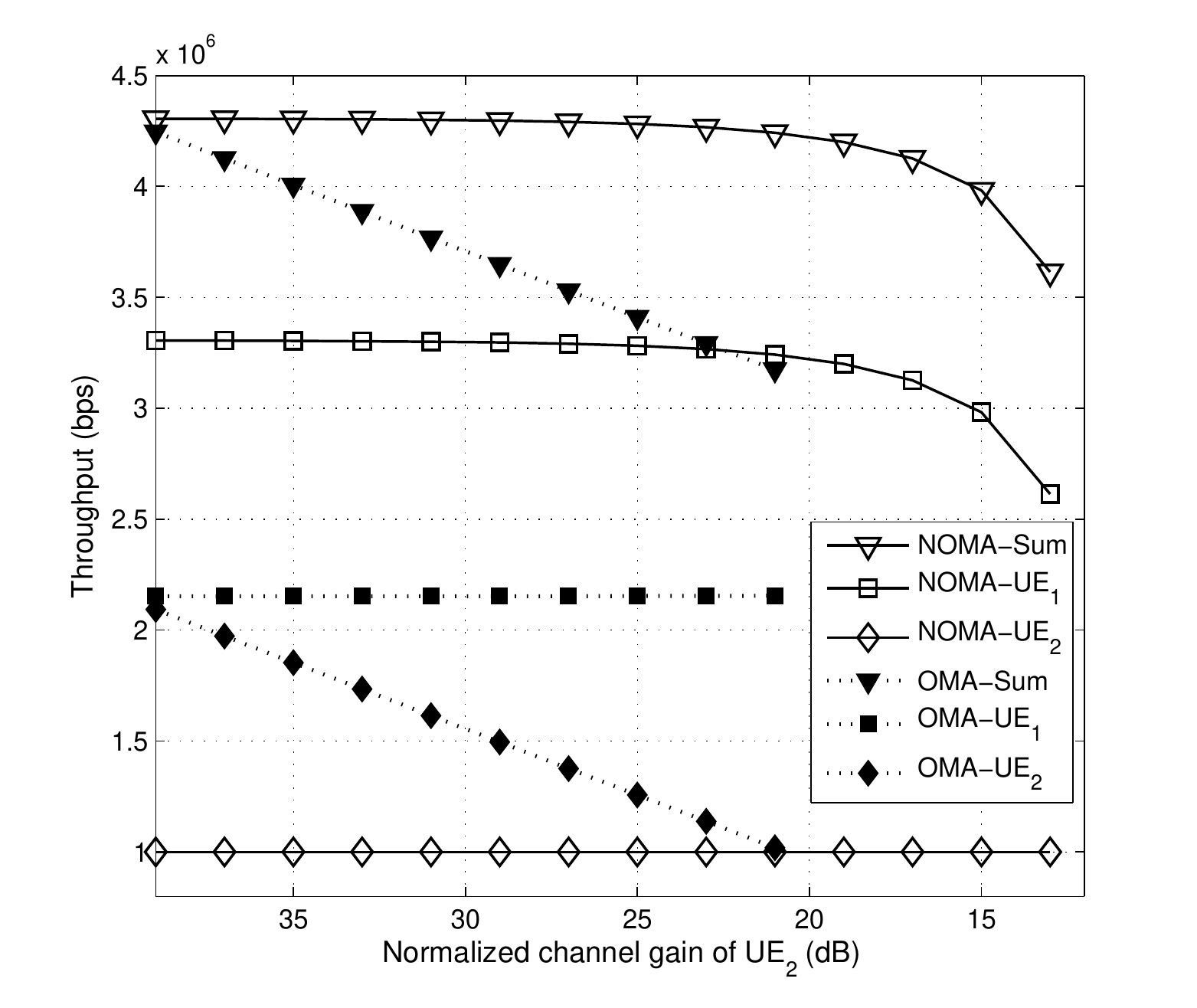}
	\caption{Throughput performance of $2$-user downlink NOMA and OMA systems assuming 1 Mbps minimum data rate. Normalized channel gain of $UE_1$ is 40 dB.}
	\label{fig:2}
 \end{center}
\end{figure}

\begin{figure}[h]
\begin{center}
	\includegraphics[width=3.6 in]{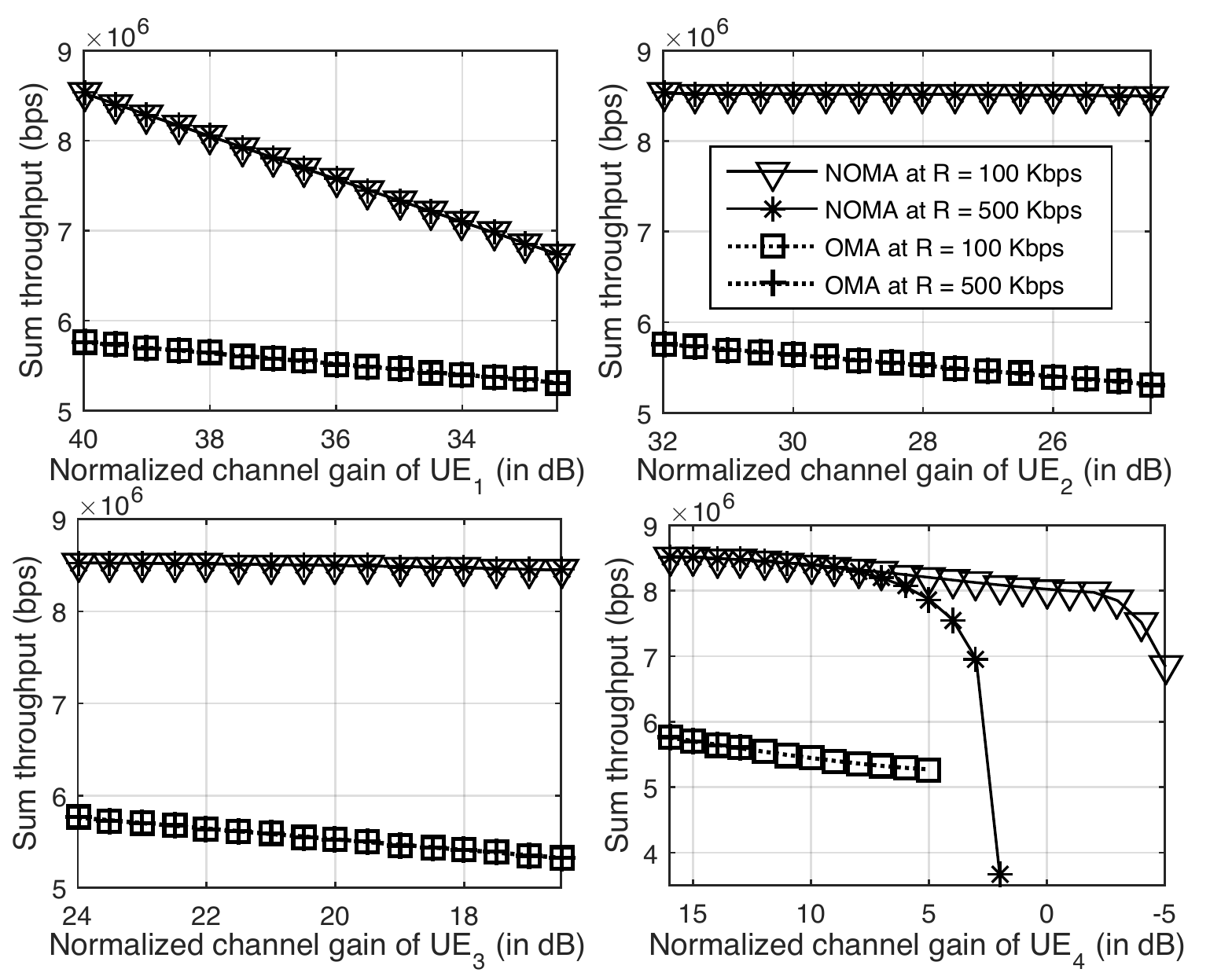}
	\caption{
Impact of channel variation on the sum-throughput of $4$-user downlink NOMA and OMA systems. Initial normalized channel gains of $UE_1$, $UE_2$, $UE_3$, and $UE_4$ are 40 dB, 32 dB, 24 dB and 16 dB, respectively.  We assume $R=$ 100 Kbps and $R=$ 500 Kbps, where $R=R_1=R_2=R_3=R_4$.
}
\label{fig:4}
 \end{center}
\end{figure}
\figref{fig:1} and \figref{fig:2} show the sum-throughput and individual throughput of $2$-user downlink NOMA cluster as well as its corresponding OMA system for minimum data rate requirements of $100$ Kbps and $1$ Mbps, respectively. The channel gain of the higher channel gain user is fixed at 40 dB whereas the channel gain of weak user varies. Further, \figref{fig:4} represents the sum-throughput of $4$-user downlink NOMA cluster and its corresponding OMA system as a function of each user channel variations. In \figref{fig:4}, the initial channel gains of $UE_1$, $UE_2$, $UE_3$, and $UE_4$ are set as $\gamma_1 = 40$ dB, $\gamma_2 = 32$ dB, $\gamma_3 = 24$ dB, and $\gamma_4 = 16$ dB, respectively, and in each sub-figure only one user's channel is varied by ensuring $\gamma_1>\gamma_2>\gamma_3>\gamma_4$. From these  simulation results we have the following observations:

\begin{itemize}
\item[\checkmark] Sum-throughput of downlink NOMA is  always better than OMA at any channel conditions. However, a significant throughput gain can be achieved  for more distinct channel conditions of users in a cluster.

\item[\checkmark]  Individual throughput of the highest channel gain user in a NOMA cluster is significantly higher than that in OMA. However, the lowest channel gain user's throughput is limited by its minimum rate requirement. To 
address this issue in NOMA,  the minimum rate requirements of different users can be dynamically adjusted to enhance the  fairness among users.

\item[\checkmark] Sum-throughput of downlink NOMA strongly depends on the  highest channel gain user within a cluster. The reason behind  is that the strongest channel gain user can cancel all interfering signals before decoding its own signal, thus its achievable data rate does not depend on inter-user interference.

\item[\checkmark] 
The impact of the lowest channel gain user is minimal on the cluster sum-throughput, unless  the channel gain is so small that a huge power is required by the lowest channel gain user. At this point, a sharp decay of sum-throughput is observed. Note that the traditional OMA is unable to operate at such poor channel gains.

\item[\checkmark]The channel variations of all users, except the highest and lowest channel gain users, in a downlink NOMA cluster do not considerably  affect the sum-throughput of a NOMA cluster (see \figref{fig:4}(b) and \figref{fig:4}(c)). 

\end{itemize}

\subsubsection{\textbf{Throughput Comparison of Various Downlink NOMA Systems}}
As given in {\bf Lemma~1}, the higher number of users in a downlink NOMA cluster significantly reduces the amount of available power for the strongest channel user, who generally contributes maximum throughput in a cluster. Therefore, it is crucial to select the correct cluster size. However, the throughput performance of a NOMA cluster depends significantly on three parameters, i.e., cluster size, transmit power, and channel gains of users. For a particular set of transmit powers and channel gains of users, we can find a cluster size that generally maximizes the sum throughput.  {\bf Table~\ref{dl-compr}} represents the sum-throughput of different NOMA and OMA systems with $12$ downlink users for various channel conditions while the transmit power is fixed in all cases.

In {\bf Table~\ref{dl-compr}}, we order the users according to their channel gains in descending  order. There are 6, 4, and 3 clusters with $2$, $3$, and $4$ users in a cluster, respectively. The throughput performances of the aforementioned downlink NOMA clusters and their respective OMA counterparts are demonstrated in {\bf Table~\ref{dl-compr}}. The main observations are as follows:
\begin{itemize}
\item[\checkmark] {\bf {\textit{Less distinct channel gains of users \textit{(cases 12 and 14)}}}}: In this case, the throughput of different downlink NOMA systems is nearly the same. However, the $4$-user NOMA achieves a higher throughput at better channel gains \textit{(case 12)}, while the $2$-user NOMA obtains higher throughput at lower channel gains \textit{(case 14)}. As such, we can conclude that higher cluster size is preferred for higher channel gains of the users and lower cluster size should be preferred  for lower channel gains of the users. The overall throughput gains of downlink NOMA over OMA are very limited.

\item[\checkmark] {\bf {\textit{More distinct channel gains of users~{(case 13):}}}} In this case, NOMA systems outperform their OMA counterparts. It can be seen that the $4$-user NOMA achieves a better throughput than $2$-user and $3$-user NOMA systems. Therefore, higher cluster size can be selected in such cases as long as the power allocation to the highest channel gain does not decrease significantly {\bf(Lemma 1)}.

\item[\checkmark]  {\bf {\textit{Number of higher channel gain users  equals to the number of  clusters~\textit{(cases 3, 4, 6):}}}} In such a case, each downlink NOMA system achieves maximum relative throughput gain compared to OMA.
In  {\bf Table~\ref{dl-compr}}, the $4$-user downlink NOMA system (i.e., $3$ clusters) achieves maximum 106.3\% throughput gain over OMA system when the number of good channels are equal to $3$  (\textit{case 3}). However,  the $3$-user (i.e., $4$ clusters) and $2$-user (i.e., $6$ clusters)  systems achieve maximum throughput gains of 86.6\% and 52.8\%, respectively, over OMA systems (\textit{case 4} and \textit{case 6}).
Thus, a NOMA system with the number of higher channel gain users equal to the number of clusters either achieves maximum or close to maximum throughput among all NOMA systems.

\item[\checkmark] In general, the throughputs of all NOMA systems are quite similar when 50\% or more users experience good channels (\textit{cases 6 to 12}). 
However, if the higher channel gain   users become limited, then the higher cluster sizes should be selected to enhance throughput (\textit{cases 1 to 3}).

\end{itemize}  

\begin{figure}[h]
\begin{center}
	\includegraphics[width=3.5 in]{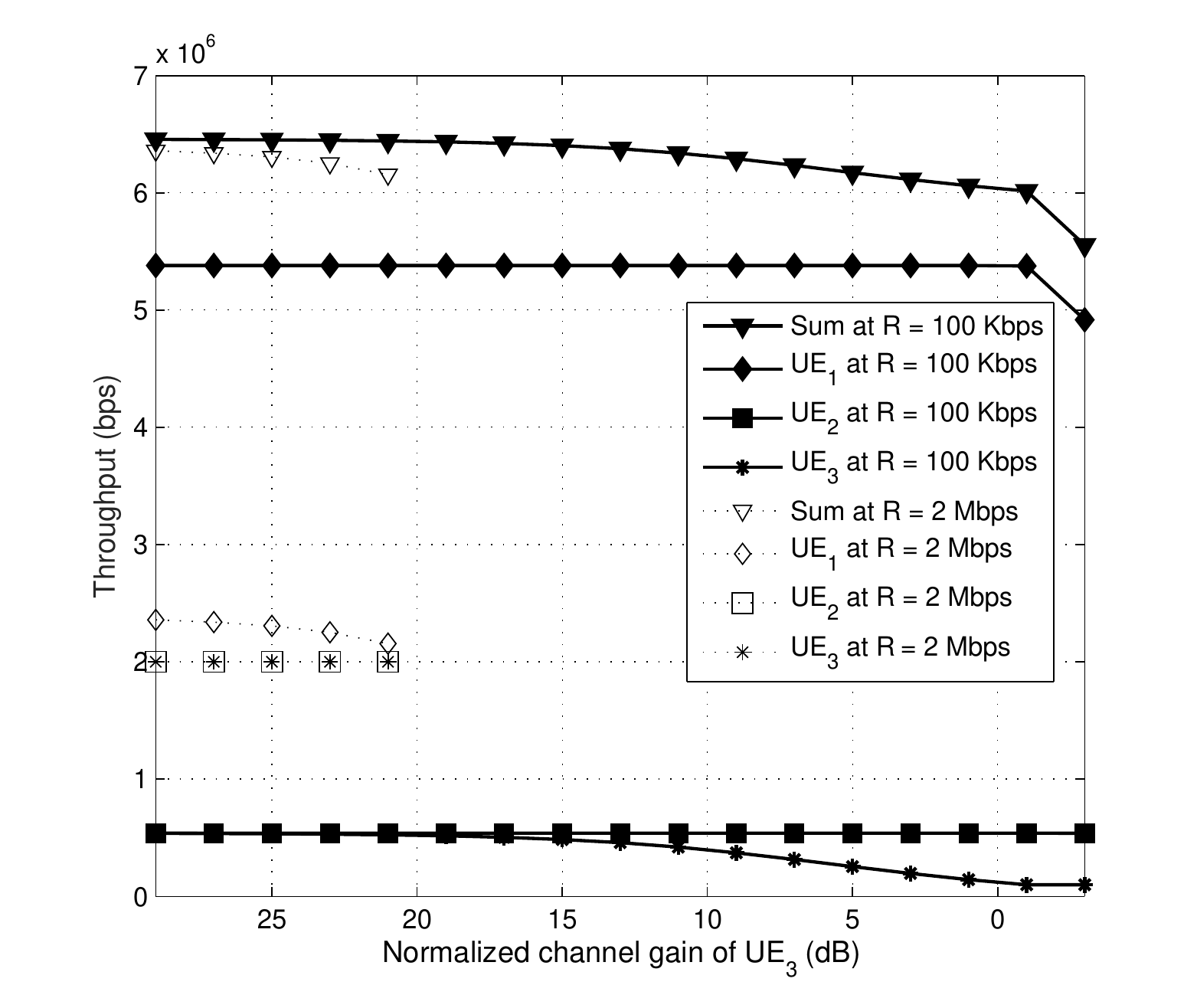}
	\caption{Throughput performances of $3$-user downlink NOMA  system, for $R=R_1=R_2=R_3=$ 100 Kbps and 2 Mbps. Normalized channel gains of $UE_1$ and $UE_2$ are 40 dB and  30 dB, respectively.}
	\label{fig:f}
 \end{center}
\end{figure}
Finally, \figref{fig:f} shows the individual users' throughput for $3$-user downlink NOMA cluster considering two different minimum rate requirements of users, i.e., 100 Kbps and 2 Mbps. In general, our optimal power allocations  maximize the transmit power of the highest channel gain user while maintaining the minimum rate requirements of all  users in a NOMA cluster. 
However, in such a case, the lower channel gain users may experience significant throughput difference when compared to the highest channel gain user, as discussed in Section~VI.A.1. As such, to improve the throughput fairness among users, the minimum  data rate requirements of the users can be adjusted further to balance the trade-off between  fairness and the overall system throughput.

\subsection{Performance Evaluation of Uplink NOMA}
In this subsection, we compare the throughput performances of uplink NOMA and OMA systems. We also compare the overall and individual user's throughput of $2$-user, $3$-user, $4$-user, and $6$-user uplink NOMA systems for $12$ uplink users. 
\\
\subsubsection{\textbf{Throughput Comparisons Between NOMA and OMA}}
The sum-throughput and individual throughput of $2$-user uplink NOMA cluster is shown in \figref{fig:6}, where the minimum user rate requirement is $100$~Kbps.  It can be seen  that the sum throughput of NOMA outperforms OMA  with more distinct channel gains of users in a cluster. Also, the NOMA sum-throughput remains  higher than OMA, regardless of the weakest channel. On the other hand, when $UE_1$ and $UE_2$ have nearly the same channel gains, $UE_2$ achieves better data rate. However, due to high interference of $UE_2$, $UE_1$ obtains very low data rate. As the channel gain of $UE_2$ reduces, its  throughput gradually reduces. However, the interference on $UE_1$ also reduces which improves the achievable throughput of $UE_1$. As such, the sum throughput of uplink NOMA cluster remains almost unchanged.

When the minimum rate requirement of both users is set as high as $1$ Mbps, \figref{fig:7} depicts the sum-throughput and individual throughput of $2$-user NOMA cluster as well as the sum throughput of $2$-user OMA system. It is observed that the sum-throughput of NOMA becomes less than the corresponding OMA system in the region when the channel gains are less distinct and power control is applied. The reason is that, without power control, the sum data rate of $2$-user uplink NOMA system and corresponding OMA system is nearly similar. Therefore, after applying power control in NOMA, sum data rate is further reduced and goes below OMA. We further note that, since the power control is only applicable at the weakest channel gain user, its impact keeps diminishing gradually for  $3$-user and beyond uplink NOMA clusters (see in \figref{fig:7new}).  \figref{fig:7new} shows the sum-throughput and individual throughput of $3$-user uplink NOMA cluster and corresponding sum throughput of OMA system, where the minimum individual rate requirement is $1$ Mbps.
It is evident that the sum-throughputs of $3$-user and beyond uplink NOMA clusters remain always better than those of the OMA systems.

\begin{figure}[h]
\begin{center}
	\includegraphics[width=3.5 in]{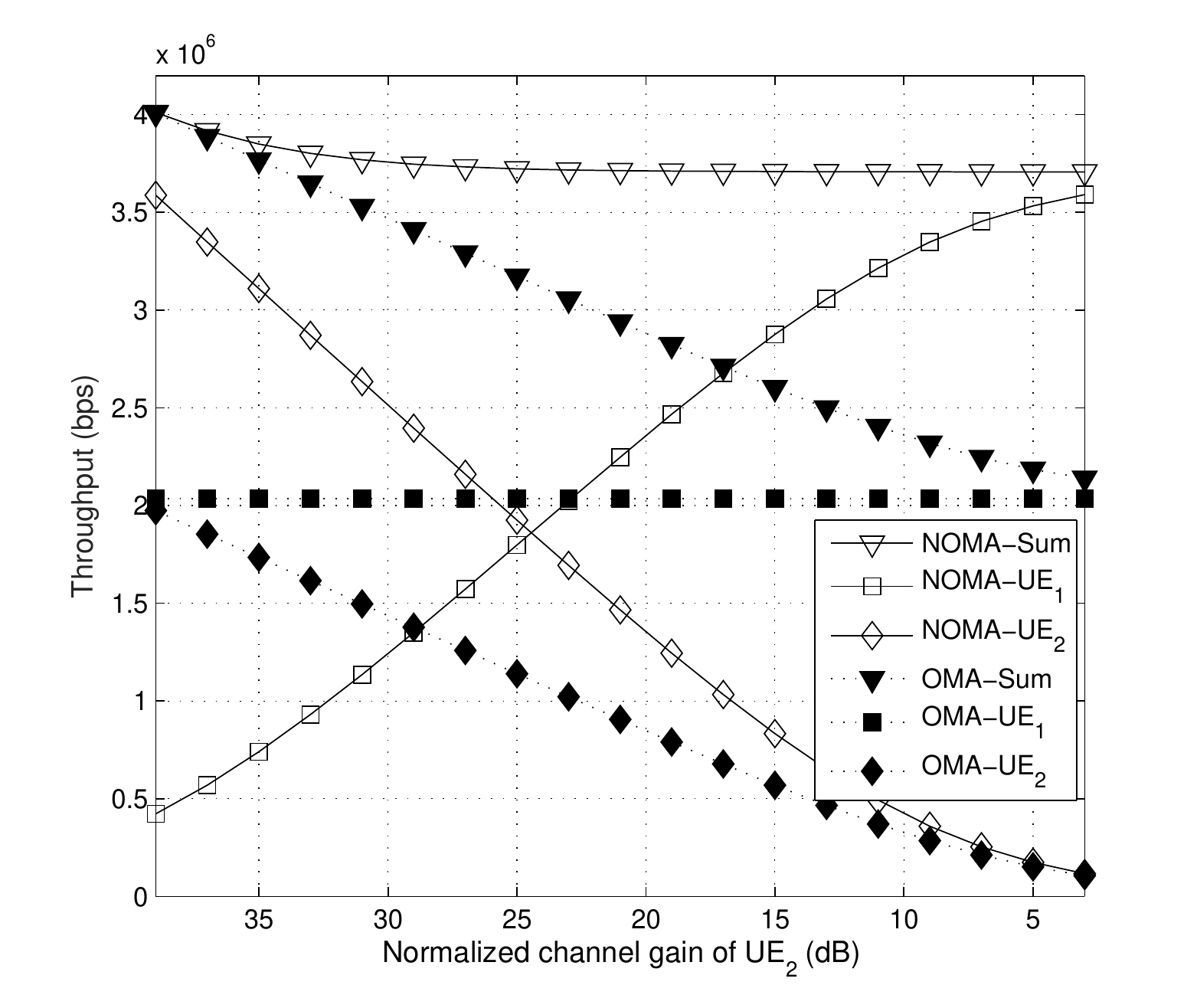}
	\caption{Throughput performance of $2$-user uplink NOMA and OMA systems assuming 100 kbps minimum data rate. Normalized channel gain of $UE_1$ is 40 dB.}
	\label{fig:6}
 \end{center}
\end{figure}
\begin{figure}[h]
\begin{center}
	\includegraphics[width=3.5 in]{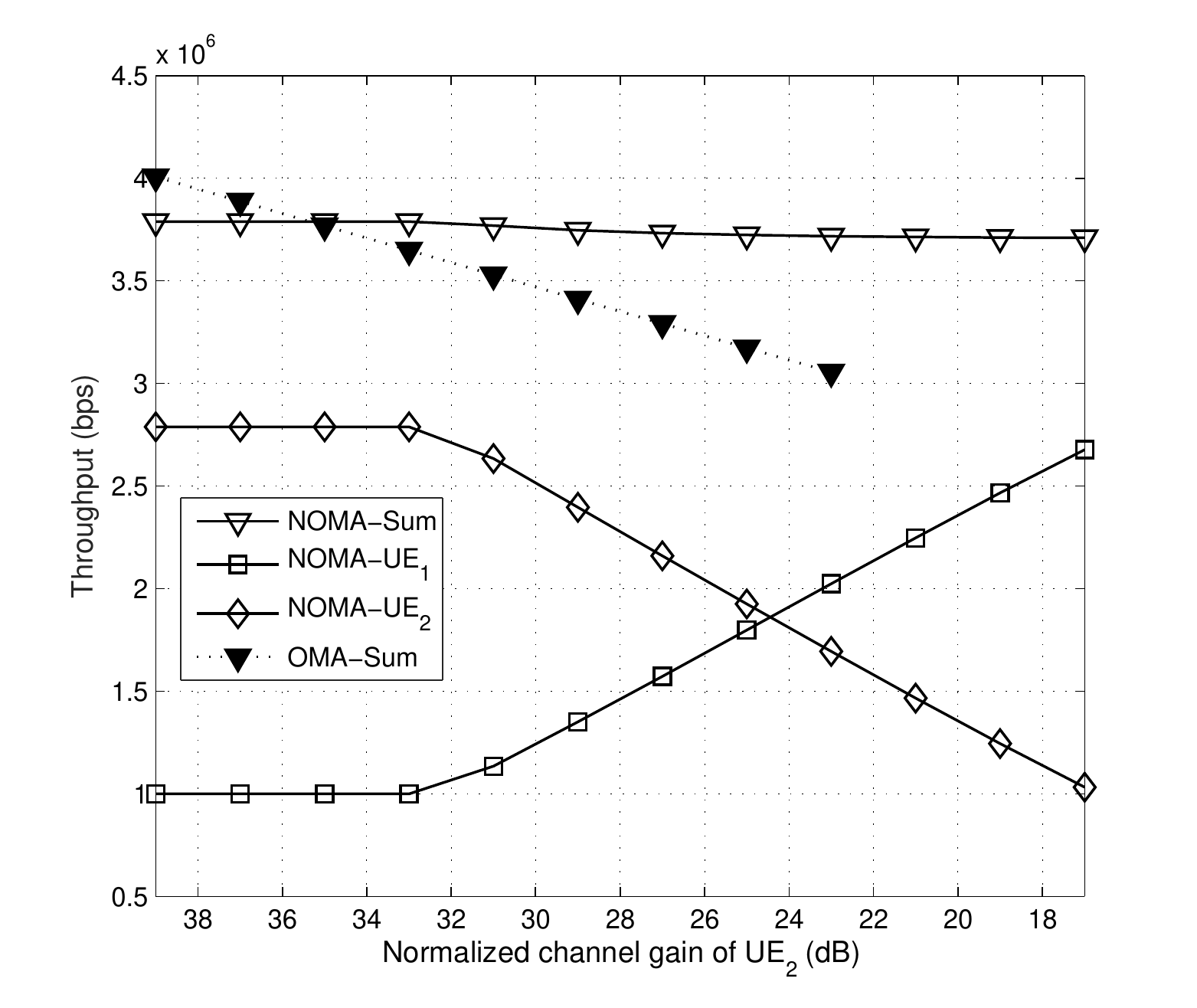}
	\caption{Throughput performance of $2$-user uplink NOMA and OMA systems assuming 1 Mbps minimum data rate. Normalized channel gain of $UE_1$ is 40 dB.}
	\label{fig:7}
 \end{center}
\end{figure}
\begin{figure}[h]
\begin{center}
	\includegraphics[width=3.5 in]{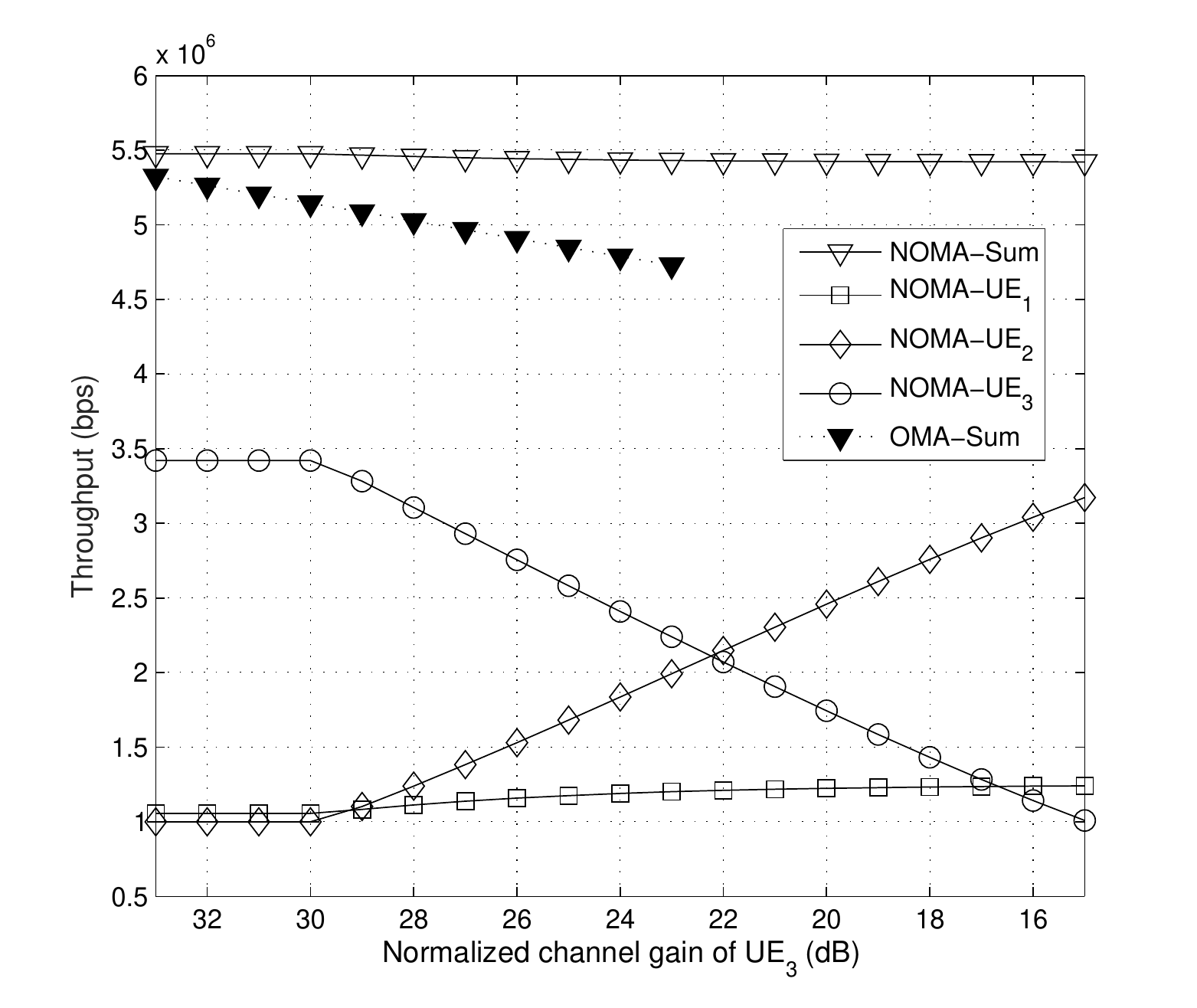}
	\caption{Throughput performance of $3$-user uplink NOMA and OMA systems assuming 1 Mbps minimum data rate. Normalized channel gains of $UE_1$ and $UE_2$ are 40 dB and 34~dB, respectively.}
	\label{fig:7new}
 \end{center}
\end{figure}
\figref{fig:8} demonstrates  the sum-throughput of $4$-user uplink NOMA cluster as a function of individual user's channel gain variation. The normalized channel gains of $UE_1$, $UE_2$, $UE_3$, and $UE_4$ are  set at 40 dB, 32 dB, 24 dB and 16 dB, respectively. In each sub-figure, the channel gain of each user is varied while others  remain fixed. For example, in case (a), the normalized channel gain of  $UE_1$ is varied from 40 dB to 32.5 dB while that channel gains of other users are fixed. It is observed that the  sum-throughput of uplink NOMA mainly depends on the channel conditions of the highest channel gain user. 

\begin{figure}[h]
\begin{center}
	\includegraphics[width=3.6 in]{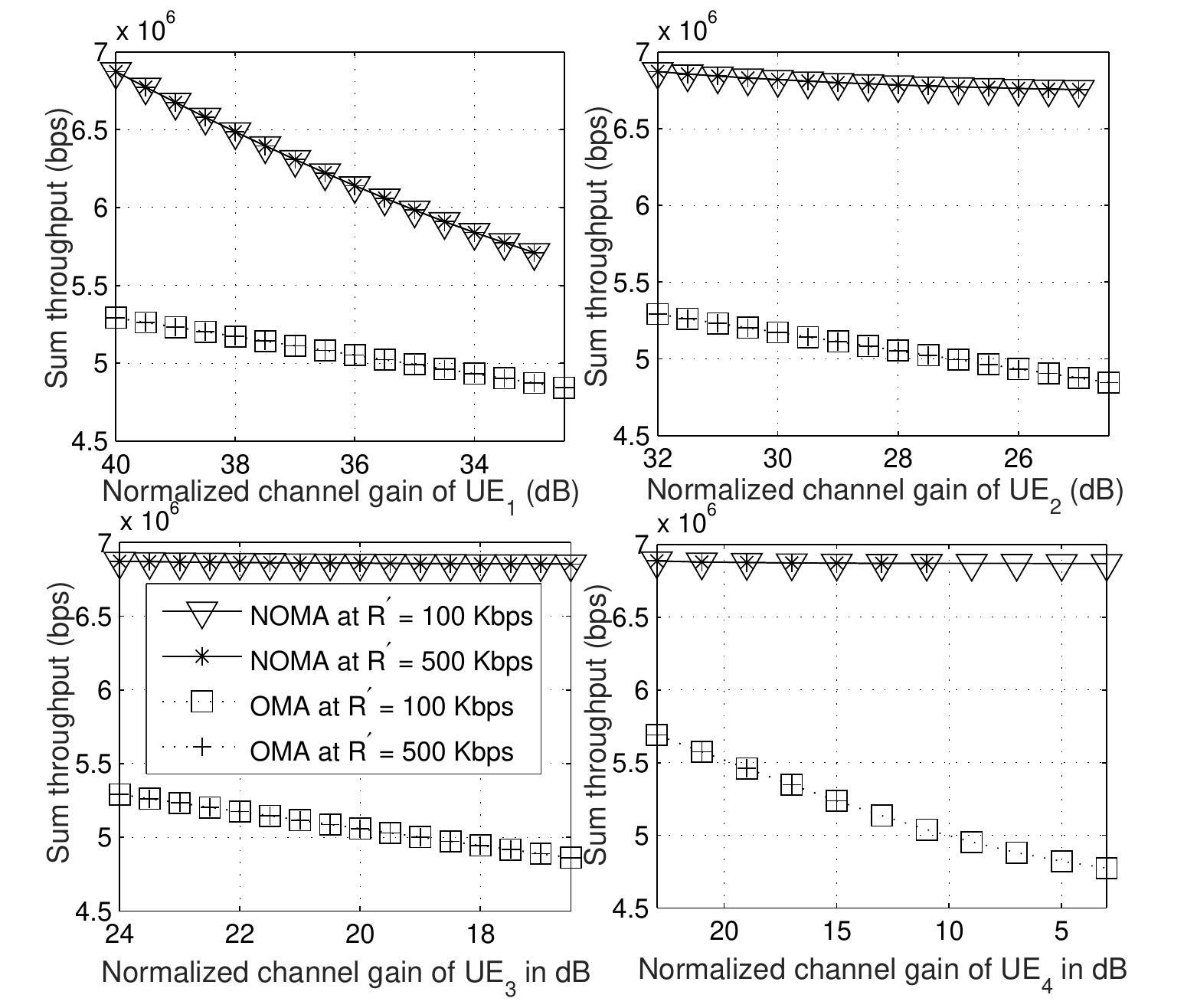}
	\caption{Impact of channel variation on the sum-throughput of $4$-user uplink NOMA and OMA systems. Initial normalized channel gains of $UE_1$, $UE_2$, $UE_3$, and $UE_4$ are 40 dB, 32 dB, 24 dB, and 16 dB, respectively.  We assume $R^\prime=$ 100 Kbps and $R^\prime=$ 500 Kbps, where $R^\prime=R_1^\prime=R_2^\prime=R_3^\prime=R_4^\prime$.}
	\label{fig:8}
 \end{center}
\end{figure}

\begin{figure}[h]
\begin{center}
	\includegraphics[width=3.6 in]{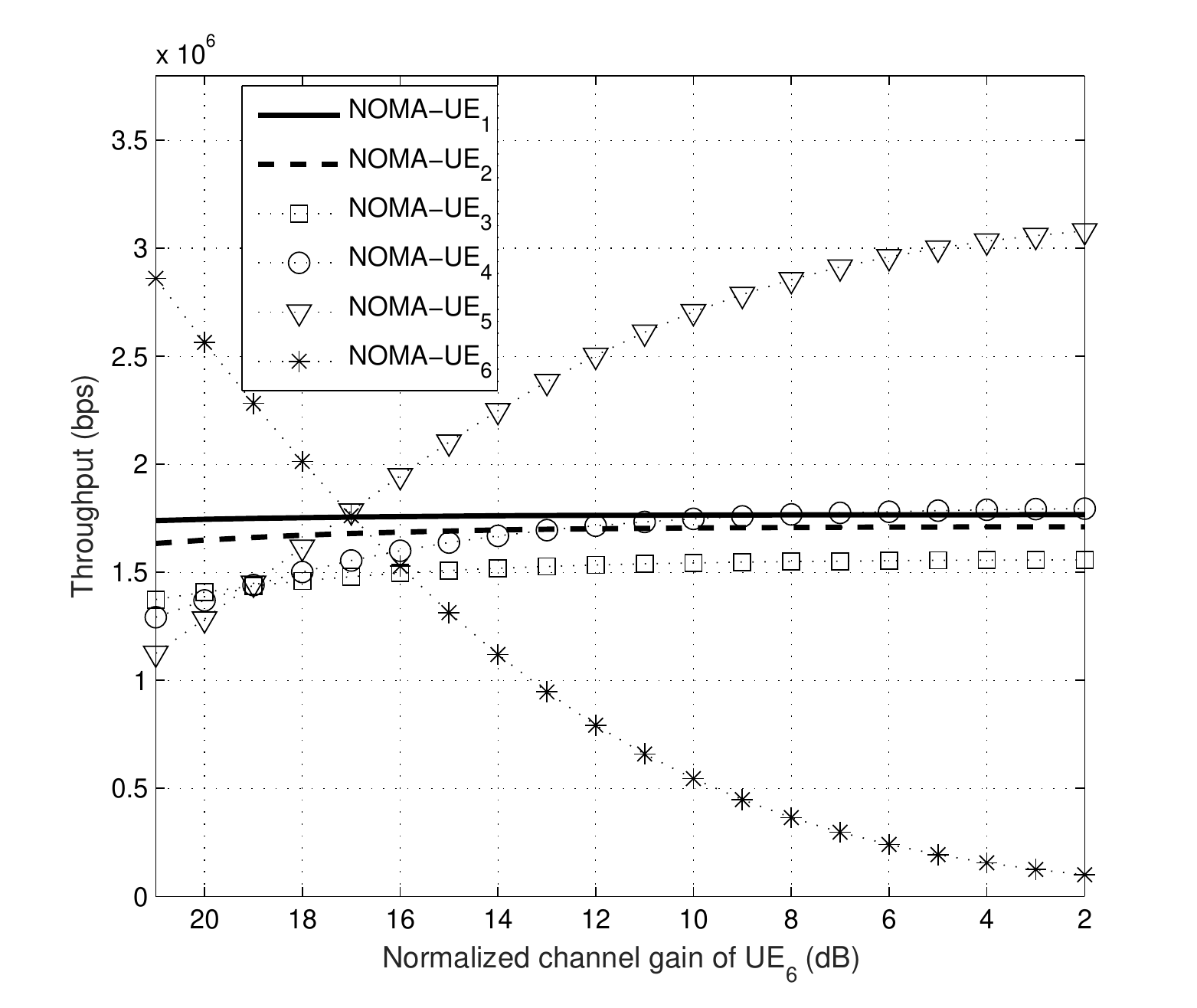}
	\caption{Throughput performance of $6$-user uplink NOMA cluster for 100 Kbps minimum data rate. Normalized channel gains of $UE_1$ to $UE_5$ are 40, 35, 30, 26, and 22~dB, respectively, while the normalized channel gain of $UE_6$ varies.}
	\label{fig:8new}
 \end{center}
\end{figure}
\
\subsubsection{\textbf{Comparisons Among Different Uplink NOMA Systems}}
The important feature of uplink NOMA system is that all lower channel gain users in a NOMA cluster interfere significantly to the higher channel gain users. Note that, due to SIC, low channel gain users do not experience any interference from high channel gain users. In contrast to downlink NOMA, the impact of transmission power control in not significant in uplink NOMA. In uplink NOMA, if more users need to be grouped into a cluster,  more distinct channels are required. With different channel conditions of $12$ uplink active users, we measure the performances of $2$-user, $3$-user, $4$-user, and $6$-user uplink NOMA systems as well as OMA system, shown in {\bf Table~\ref{ul-compr}}. The users are clustered according to the method discussed in Section~IV.C. The channel gains in  Table {\bf\ref{ul-compr}} are chosen in order to ensure the minimum channel distinctness required for $6$-user uplink NOMA system.

In Table~{\bf\ref{ul-compr}}, we sort 12  users according to the descending order of their channel gains. There are 6, 4, 3, and 2 clusters available for $2$-user, $3$-user, $4$-user, and $6$-user uplink NOMA systems, respectively. From Table~\ref{ul-compr}, we have the following key observations: 
\begin{itemize}
\item[\checkmark] {\bf {\textit{ More distinct channel gains of users~{(case 13):}}}} In this case, the performance of any uplink NOMA is much better than that of the OMA system (\textit{case 13}), whereas the higher order NOMA achieves higher throughput.

\item[\checkmark] {\bf {\textit{Less distinct channel gains of users \textit{(cases 12 and 14)}}}}: In this case, higher order uplink NOMA still performs better, although the throughput performance of various NOMA systems are quite similar. However, in this case, the throughput gains of NOMA over OMA are marginal. 

\item[\checkmark]  {\bf {\textit{Number of good channel users  equals the number of  clusters~\textit{(cases 2, 3, 4, 6):}}}} In such a case,  each uplink NOMA system achieves maximum relative throughput gain compared to OMA.
In {\bf Table~\ref{ul-compr}},  \textit{case 2} shows 105.4\% throughput gain of $6$-user uplink NOMA system in comparison to OMA system, while \textit{case 3, 4} and \textit{case 6} show the maximum throughput gain of $4$-user, $3$-user, and $2$-user uplink NOMA systems, respectively, where the gains are 79.9\%, 60.25\%, and 33.3\%, respectively. Thus, a NOMA system with the number of higher channel gain users equal to the number of clusters either achieves maximum or close to maximum throughput among all NOMA systems.


\item[\checkmark] When one or couple of users experience higher channel gains, higher order uplink NOMA systems perform much better than OMA systems \textit{(case 1, 2 and 3)}. As the number of  users experiencing higher channel gains increases, different uplink NOMA systems perform nearly similar \textit{(case 4 to 12)}.
\end{itemize}  
Fig. \ref{fig:8new} shows the individual throughput of $6$-user uplink NOMA cluster, where $UE_1$ to $UE_5$ experience 40, 35, 30, 26, and 22 dB channel gains, respectively.  Fig. \ref{fig:8new} shows good throughput fairness among $UE_1$ to $UE_4$, whilst, by selecting  proper channel gains for $UE_5$ and $UE_6$, all six users can get good throughput fairness.  Therefore, exploitation of channel gain differences among the NOMA users is the key to designing efficient uplink NOMA systems.

\section{Conclusion}
Efficient user clustering and power allocation among  NOMA users are the key design issues for successful NOMA operations.
In this paper, for both uplink and downlink NOMA, we have formulated a joint optimization problem for sum throughput maximization under the constraints of transmission power budget, minimum rate requirements of users, and SIC receiver's operation constraints. Due to the combinatorial nature of the problem, we have developed a low-complexity sub-optimal user clustering scheme. Given the user clustering, we have derived closed-form optimal power allocations for  $m$-user uplink/downlink NOMA systems.   In both of downlink and uplink NOMA, user clusters with more distinctive channel gains provides impressive throughput gain of NOMA over their counterpart OMA systems. Numerical results show that the performance of downlink NOMA decreases if the cluster size increases beyond a certain threshold. Despite  the numerous benefits of NOMA over OMA, the issues such as SIC error processing and inter-cell interference, are still under investigation. In this paper, we have considered ideal SIC; however, the performance of NOMA may depend significantly on the SIC errors. In  downlink NOMA, each signal needs to be identically encoded, modulated, and precoded at the BS while the SIC receiver needs to successively demodulate, decode, and re-modulate stronger signals. Therefore, for large NOMA clusters, error propagation in SIC may drastically reduce the NOMA performance.  Further, inter-cell interference is another major challenge for practical NOMA applications. The power allocation solutions need to be carefully extended for dense cellular networks in which inter-cell interference can be significant.

\appendices
\section*{Appendix~A: Proof of Satisfied KKT Condition for  Downlink NOMA}

\renewcommand{\theequation}{A.\arabic{equation}}
\setcounter{equation}{0}

In this Appendix, we show the verification of KKT conditions for a given  Lagrange multiplier combination in a general $m$-user downlink NOMA cluster considering a particular value of $m$. Let us consider $m=4$ which is a $4$-user downlink NOMA cluster. Then, the set of Lagrange multipliers can be expressed as $A = \{\lambda\}, \, B = \{\mu_1, \mu_2, \mu_3, \mu_4\} \, \text{and} \, C = \{\psi_2, \psi_3, \psi_4\}$. The solution sets that satisfy KKT conditions are, $S = \{\lambda, \mu_2\,\text{or}\,\psi_2, \mu_3\,\text{or}\,\psi_3, \mu_4\,\text{or}\,\psi_4\}$. Now, consider one set of Lagrange multipliers, say $S_1 = \{\lambda, \mu_2, \mu_3, \mu_4\}$, that satisfies the KKT conditions, and thus the values of the Lagrange multipliers are $\lambda^\ast > 0$, $\mu_2^\ast > 0$, $\mu_3^\ast > 0$, $\mu_4^\ast > 0$, and  $\mu_1^\ast = \psi_2^\ast = \psi_3^\ast = \psi_4^\ast =0$.

For the aforementioned Lagrange multipliers \eqref{opd4} to (\ref{opd6}) can be expressed as 
\begin{align}
P_t-\sum\limits_{i=1}^{4}P_i = 0   
\label{eqn:a1}
\end{align}
\begin{align}
P_i\gamma_i-\Big(\sum\limits_{j=1}^{i-1}P_j\gamma_i+\omega\Big)\Big(\varphi_i-1\Big) = 0, \, \forall\, i = 2,3,4, 
\label{eqn:a2}
\end{align}
\begin{align}
P_i\gamma_i-\Big(\sum\limits_{j=1}^{i-1}P_j\gamma_i+\omega\Big)\Big(\varphi_i-1\Big) > 0, \, \forall\, i = 1, 
\label{eqn:a3}
\end{align}
\begin{align}
P_i\gamma_{i-1}-\sum\limits_{j=1}^{i-1}P_j\gamma_{i-1} - P_{tol} > 0, \, \forall\, i = 2,3,4.
\label{eqn:a4}
\end{align} 
Equations (\ref{eqn:a1})-(\ref{eqn:a2}) provide the optimal solution of $P_1, P_2, P_3$ and $P_4$, while equations (\ref{eqn:a3})-(\ref{eqn:a4}) provide the necessary conditions of these optimal solutions. Now, solving the equations (\ref{eqn:a1})-(\ref{eqn:a2}), we obtain the optimal power allocations as follows:
\begin{align}
P_1 = \frac{P_t}{\varphi_2\varphi_3\varphi_4} - \frac{\omega(\varphi_2-1)}{\varphi_2\gamma_2} - \frac{\omega(\varphi_3-1)}{\varphi_2\varphi_3\gamma_3} - \frac{\omega(\varphi_4-1)}{\varphi_2\varphi_3\varphi_4\gamma_4}, \nonumber 
\end{align}
\begin{align}
P_2 =& \frac{P_t(\varphi_2-1)}{\varphi_2\varphi_3\varphi_4} +  \frac{\omega(\varphi_2-1)}{\varphi_2\gamma_2}- \frac{\omega(\varphi_2-1)(\varphi_3-1)}{\varphi_2\varphi_3\gamma_3} - \nonumber\\ 
&\frac{\omega(\varphi_2-1)(\varphi_4-1)}{\varphi_2\varphi_3\varphi_4\gamma_4}, \nonumber
\end{align}
\begin{align}
P_3 = \frac{P_t(\varphi_3-1)}{\varphi_3\varphi_4} + \frac{\omega(\varphi_3-1)}{\varphi_3\gamma_3} - \frac{\omega(\varphi_3-1)(\varphi_4-1)}{\varphi_3\varphi_4\gamma_4}, \nonumber 
\end{align}
\begin{align}
P_4 = \frac{P_t(\varphi_4-1)}{\varphi_4} + \frac{\omega(\varphi_4-1)}{\varphi_4\gamma_4}. \nonumber
\end{align}

Since $\varphi_i =2^{\frac{R_1}{\omega B}} > 1,\,\forall\,i=1,2,3,4$, the solutions of $P_1,P_2, P_3$ and $P_4$ all are positive. Now, with the considered Lagrange multipliers and positive $P_i,\,\forall\,i=1,2,3,4$, and after some algebraic operations, equations (\ref{opd2}) to (\ref{opd3}) can be expressed as
\begin{align}
\mathlarger{\mathlarger{‎‎\sum}}_{j=1}^{3}&\frac{\omega^2B(\gamma_j - \gamma_{j+1})}{\bigg(\sum\limits_{k=1}^{j}P_k\gamma_j+\omega\bigg)\bigg(\sum\limits_{k=1}^{j}P_k\gamma_{j+1}+\omega\bigg)} + \frac{\omega B\gamma_4}{\sum\limits_{l=1}^{4}P_l\gamma_4+\omega}, \enspace \nonumber \\
&= \lambda + \sum\limits_{j=2}^{4}(\varphi_j-1)\mu_j\gamma_j,
\label{eqn:a5}
\end{align}
\begin{align}
\mathlarger{\mathlarger{‎‎\sum}}_{j=i}^{3}&\frac{\omega^2B(\gamma_j - \gamma_{j+1})}{\bigg(\sum\limits_{k=1}^{j}P_k\gamma_j+\omega\bigg)\bigg(\sum\limits_{k=1}^{j}P_k\gamma_{j+1}+\omega\bigg)} + \frac{\omega B\gamma_4}{\sum\limits_{l=1}^{4}P_l\gamma_4+\omega} \enspace \nonumber \\
&= \lambda - \mu_i\gamma_i + \sum\limits_{j=i+1}^{4}(\varphi_j-1)\mu_j\gamma_j, \, \forall \, i =2,3,4.
\label{eqn:a6}
\end{align}

After performing some algebraic operations into equations (\ref{eqn:a5})-(\ref{eqn:a6}), we obtain the Lagrange multipliers as follows:
\begin{align}
\mu_2 = \frac{\omega^2 B(\gamma_1-\gamma_2)}{\varphi_2\gamma_2 (P_1\gamma_1+\omega)(P_1\gamma_2+\omega)}, 
\label{eqn:a7}
\end{align}
\begin{align}
\mu_i = \frac{\omega^2 B(\gamma_{i-1}-\gamma_i)}{\varphi\gamma_i \Big(\sum\limits_{j=1}^{i-1}P_j\gamma_{i-1}+\omega\Big)\Big(\sum\limits_{j=1}^{i-1}P_j\gamma_{i}+\omega\Big)} + \nonumber \\
\mu_{i-1}\gamma_{i-1}, \, \forall \, i =3,4,
\label{eqn:a8}
\end{align}
\begin{align}
\lambda = \frac{\omega B\gamma_4}{\sum\limits_{j=1}^{4}P_j\gamma_4+\omega} + \mu_4\gamma_4. 
\label{eqn:a9}
\end{align}

In our proposed dynamic power allocation solutions, we sort the UEs according to the descending order of their normalized channel gains, i.e $\gamma_1 > \gamma_2 > \gamma_3 > \gamma_4$. Therefore, the solutions for $\lambda$ and $\mu_i, \, \forall \, i = 2,3,4$, of equation (\ref{eqn:a7})-(\ref{eqn:a9}) all are positives. Therefore, the set of Lagrange multipliers $S_1 = \{\lambda, \mu_2, \mu_3, \mu_4\}$ satisfies the KKT conditions. All the other cases can easily be verified by using  a similar approach.


\section*{Appendix-B: Proof of Satisfied KKT Conditions for  Uplink NOMA}
\renewcommand{\theequation}{B.\arabic{equation}}
\setcounter{equation}{0}

Similar to downlink NOMA cluster, we show the KKT verification of a general $m$-user uplink NOMA cluster by considering a particular combination of Lagrange multipliers in a particular cluster size. Let again consider $m=4$, i.e., a $4$-user uplink NOMA cluster. Then, the sets of Lagrange multipliers can be expressed as $A = \{\lambda_1, \lambda_2, \lambda_3, \lambda_4\}, \, B = \{\mu_1, \mu_2, \mu_3, \mu_4\} \, \text{and} \, C = \{\psi_1, \psi_2, \psi_3\}$. Therefore, the solution sets that satisfy the KKT conditions are: $S = \{\lambda_1, \lambda_2, \lambda_3, \lambda_4\,\text{or}\,\mu_3\,\text{or}\,\psi_3\}$. Now, consider one set of Lagrange multipliers, say, $S_1 = \{\lambda_1, \lambda_2, \lambda_3, \mu_3\}$, satisfies the KKT conditions, and thus the value of the Lagrange multipliers are $\lambda_i^\ast > 0$, $\mu_3^\ast > 0$, and  $\lambda_4^\ast = \mu_j^\ast=\psi_k^\ast=0, \,\forall \, i = 1,2,3, \, \forall \, j = 1,2,4, \, \forall \,  k = 1,2,3$.

Now, with the aforementioned Lagrange multipliers, equations (\ref{opu4})-(\ref{opu6}) can be expressed as
\begin{align}
P_t^\prime-P_i = 0, \, \forall \, i = 1,2,3,
\label{equ:b1}
\end{align}
\begin{align}
P_3\gamma_3- \phi_3P_4\gamma_4 - \phi_3 \omega = 0 ,
\label{equ:b2}
\end{align}
\begin{align}
P_t^\prime-P_4 > 0, 
\label{equ:b3}
\end{align}
\begin{align}
P_i\gamma_i-\phi_i \sum\limits_{j = i+1}^{4}P_j\gamma_j- \phi_i \omega > 0, \, \forall \, i > 1,2,4,
\label{equ:b4}
\end{align}
\begin{align}
P_i \gamma_i - \sum\limits_{j = i+1}^{4}P_j\gamma_j - P_{tol} > 0,  \quad \forall \, i = 1,2,3.
\label{equ:b5}
\end{align}
Equations (\ref{equ:b1})-(\ref{equ:b2}) provide the optimal solution of $P_1, P_2, P_3$ and $P_4$, while equations (\ref{equ:b3})-(\ref{equ:b5}) provide the necessary conditions of these optimal solutions. Now, solving equations (\ref{equ:b1})-(\ref{equ:b2}), we obtain the optimal power allocations as follows:
\begin{align*}
P_i &= P_t^\prime, \,\forall \, i = 1,2,3,
\end{align*}
\begin{align*}
P_4 &= \frac{P_t^\prime\gamma_3}{\phi_3\gamma_4} - \frac{\omega}{\gamma_4}.
\end{align*}

Since $\varphi_i =2^{\frac{R^\prime_1}{\omega B}}-1 > 0,\,\forall\,i=1,2,3,4$, the solution of $P_1,P_2, P_3$ and $P_4$ all are positive. Now, with the considered Lagrange multipliers and the resultant positive $P_i,\,\forall\,i=1,2,3,4$, equations (\ref{opu2})-(\ref{opu3}) can be expressed as follows:
\begin{align}
\lambda_i = \frac{\omega B\gamma_i}{\sum\limits_{j=1}^{4}P_j\gamma_j+\omega}, \, \forall \, i = 1,2,
\label{equ:b6}
\end{align}
\begin{align}
\mu_3 = \frac{\omega B}{\phi_3\sum\limits_{j=1}^{4}P_j\gamma_j+\phi_3 \omega}, 
\label{equ:b7}
\end{align}
\begin{align}
\lambda_3 = \frac{\omega B\gamma_3}{\sum\limits_{j=1}^{4}P_j\gamma_j+\omega} + \gamma_3\mu_3.
\label{equ:b8}
\end{align}

Since the solutions of $P_i$, $\forall \, i =1,2,3,4$, all  are positive, equations (\ref{equ:b6})-(\ref{equ:b8}) show that the $\lambda_1$, $\lambda_2$, $\lambda_3$, and $\mu_3$ all are positive. Therefore, the Lagrange multiplier set $S_1 = \{\lambda_1, \lambda_2, \lambda_3, \mu_3\}$, satisfies the KKT conditions. All the other cases can easily be verified by using a similar approach.

\bibliographystyle{IEEE}

\end{document}